\documentclass[journal,10pt]{IEEEtran}
\makeatletter
\def\endthebibliography{% 
    \def\@noitemerr{\@latex@warning{Empty `thebibliography' environment}}%
    \endlist
}
\makeatother

\usepackage{cite}

\usepackage[pdftex]{graphicx}
\usepackage[caption=false,font=footnotesize]{subfig}
\usepackage{tikz}
\usetikzlibrary{arrows.meta, positioning, shapes.geometric}
\usepackage{amsmath}
\usepackage{mathtools, cuted}
\usepackage{amssymb}
\usepackage{bm}
\usepackage{mathrsfs}
\usepackage{breqn}

\usepackage{pifont}
\usepackage{xcolor}
\usepackage{url}
\usepackage{lettrine}
\usepackage{lipsum}
\usepackage{siunitx}
\usepackage{soul}
\usepackage{array}
\usepackage[inline]{enumitem}
\usepackage{epsfig}
\usepackage{filecontents}
\usepackage{booktabs}

\usepackage{algpseudocode}
\usepackage{algorithm, tabularx}
\usepackage{multirow}
\newcolumntype{L}[1]{>{\raggedright\let\newline\\\arraybackslash\hspace{0pt}}m{#1}}
\newcolumntype{C}[1]{>{\centering\let\newline\\\arraybackslash\hspace{0pt}}m{#1}}
\newcolumntype{R}[1]{>{\raggedleft\let\newline\\\arraybackslash\hspace{0pt}}m{#1}}
\newlength{\maxwidth}

\makeatletter
\newcommand{\multiline}[1]{%
	\begin{tabularx}{\dimexpr\linewidth-\ALG@thistlm}[t]{@{}X@{}}
		#1
	\end{tabularx}
}
\makeatother
\algdef{SE}[SUBALG]{Indent}{EndIndent}{}{{\algorithmicend\ }}
\algtext*{Indent}
\algtext*{EndIndent}

\usepackage{amsthm}
\newtheorem{theorem}{Theorem}

\newtheorem{lemma}[theorem]{Lemma}
\newtheorem{corollary}{Corollary}[theorem]
\theoremstyle{remark}

\newtheorem{remark}{Remark}
\DeclareMathOperator*{\argmin}{\arg\min}
\DeclareMathOperator*{\argmax}{\arg\max}
\newcommand{\Conv}{\mathrm{Conv}}

\begin{document}

\title{Beamforming Gain Maximization for\\Fluid Reconfigurable Intelligent Surface:\\A Minkowski Geometry Approach}

\author{Hong-Bae Jeon,~\IEEEmembership{Member,~IEEE}%
\thanks{\textit{(Corresponding Author: Hong-Bae Jeon)}}%
\thanks{H.-B. Jeon is with the School of Electronic Engineering, Soongsil University, Seoul, Korea (e-mail: hongbae08@ssu.ac.kr).}%
}

\maketitle
\begin{abstract}
This paper studies joint beamforming and configuration optimization for a fluid reconfigurable intelligent surface (FRIS)-assisted downlink system with port selection and finite-resolution phase shifts. We develop an alternating-optimization (AO) framework that combines a closed-form maximum-ratio transmission update with a globally optimal FRIS update for a fixed beamformer, and converges monotonically within finitely many iterations. The key idea is to represent the feasible reflected-signal set as a union of Minkowski sums over port subsets and to characterize its maximum magnitude through support functions. This reduces the FRIS subproblem to a one-dimensional directional search, in which the optimal configuration is obtained by selecting the ports with the largest direction-dependent scores together with their best phase codewords. For uniformly quantized phase codebooks, closed-form scores and a finite set of critical angles are derived, identifying the global optimum with near-quadratic complexity. Numerical results demonstrate superior performance over existing benchmarks, near-optimal beamforming gains, robustness to channel uncertainty, and rapid convergence.
\end{abstract}


\IEEEpeerreviewmaketitle
\section{Introduction}
\label{sec:intro}

\lettrine{T}{he} continuing evolution of wireless communications is being accelerated by an unprecedented demand for extremely high data rates, ultra-reliable transmission, and low-latency connectivity, driven by the proliferation of data-intensive and latency-critical services~\cite{NewHor, smidaFD}. To accommodate these stringent requirements, sixth-generation (6G) wireless systems are envisioned to transcend traditional link-level design and adopt new paradigms that simultaneously address spectral-efficiency~\cite{yhFD}, interference mitigation in ultra-dense deployments~\cite{MA}, and communication reliability under high mobility and adverse propagation environments~\cite{hjcoop, hjaib}. Meanwhile, these ambitious performance targets must be achieved under tight constraints on energy consumption and hardware complexity, necessitating fundamentally new approaches to wireless system architecture and signal processing.

In this context, the ability to actively reconfigure the wireless propagation environment has emerged as a key enabler for 6G networks~\cite{renzosmart}. Rather than treating the radio environment as an uncontrollable and often hostile medium, future wireless systems are envisioned to incorporate intelligent, environment-aware components that can adaptively shape signal propagation to improve coverage, reliability, and energy-efficiency~\cite{alexsmartris}. This paradigm shift has naturally led to the development of reconfigurable intelligent surfaces (RIS), which provide a scalable and cost-effective means of manipulating electromagnetic waves through the electronic control of a large number of low-cost reflecting elements~\cite{risspm}.

By appropriately adjusting the phase responses, RIS-assisted systems can coherently combine the reflected signals with the direct path, thereby enhancing the received signal strength and coverage without requiring additional radio-frequency chains at the surface~\cite{nfris}. Owing to this appealing property, RIS have attracted considerable attention as an energy- and cost-efficient alternative to conventional relaying architectures~\cite{vsrelay,HBRIS22}. Despite these advantages, conventional RIS architectures are inherently passive and rely on reflecting elements with fixed physical locations. As a result, their achievable performance gains are fundamentally constrained by the severe double-fading attenuation of cascaded channels~\cite{DF}, particularly under unfavorable propagation conditions such as blockage or large path loss~\cite{HBRIS,HBRIS22}. Moreover, the lack of position reconfigurability limits the ability of RIS to fully adapt to channel variations, even when sophisticated phase optimization techniques are employed.

To overcome the rigidity of fixed-position RIS, fluid reconfigurable intelligent surfaces (FRIS) have been introduced~\cite{FRISlook,FRISmag}, inspired by fluid~\cite{FAS} and reconfigurable~\cite{trimimo22} antenna systems. In FRIS, the surface is equipped with a dense grid of candidate ports, among which only a subset is dynamically selected. This port-selection capability provides an additional spatial degree-of-freedom (DoF) beyond phase control, enabling the system to exploit location diversity and to form more favorable cascaded channels by selecting effective ports from a large candidate set~\cite{FRISonoff, FRISsec}. {This coupled propagation-level design distinguishes FRIS from conventional RIS and FAS architectures, where phase adaptation and aperture/port selection on propagation level are typically treated under different signal models.} Therefore, FRIS has attracted growing research interest, and a variety of studies have been conducted to investigate its performance limits, practical architectures, and optimization strategies under different system settings. {It is worth noting that, although large-scale experimental validation of FRIS itself remains limited to date, the underlying concept is supported by recent advances in fluid antenna systems and reconfigurable surfaces, where dynamic port activation and discrete phase control have been experimentally demonstrated~\cite{r1r1, r2r2, fasliq, fasex2}. In this sense, FRIS can be interpreted as a practically motivated abstraction that captures partial port selection and finite-resolution phase adaptation, rather than an idealized fully reconfigurable surface.}

In early studies on FRIS, the authors of~\cite{FRISlook} and~\cite{FRISpa} provided initial insights into the performance analysis and enhancement potential of FRIS by introducing position reconfigurability. Specifically,~\cite{FRISlook} demonstrated notable performance gains over conventional RIS in both single and multi-user scenarios through joint optimization framework, while~\cite{FRISpa} presented analytical characterizations of outage probability and capacity with its tight upper-bound under statistical channel models. Building on this concept,~\cite{FRISonoff} investigated practical FRIS architectures with on/off port selection and discrete phase shifts, showing that selecting a subset of ports can significantly improve link performance while reducing hardware overhead. Recently,~\cite{FRISambc} explored FRIS in ambient backscatter communication systems, where optimizing the positions of fluid elements was shown to effectively mitigate the double-fading effect of cascaded links and significantly improve the achievable backscatter rate compared to conventional RIS-based designs. The potential of FRIS in secure communications was further explored in~\cite{FRISsec} and~\cite{FRISsec22}, where FRIS-assisted systems were shown to substantially enhance secrecy performance under spatial correlation and partial selection of ports. More recently, extensions such as element-level pattern reconfigurability~\cite{FRISbp} have been proposed to further enlarge the design space by jointly optimizing radiation patterns and beamforming. Together, these works clearly demonstrate the performance gains achievable by FRIS.

However, practical FRIS implementations are inevitably subject to stringent hardware constraints, under which each selected port can typically apply only a finite-resolution, unit-modulus phase coefficient~\cite{FRISonoff, FRISbp}. Under these constraints, FRIS configuration gives rise to a mixed discrete optimization over the selected-port index set and the discrete phase codewords, which is further coupled with the multi-antenna base-station (BS) beamformer through the resulting effective channel gain~\cite{JAP}. Owing to this strong coupling, direct joint optimization is generally intractable, and a variety of algorithmic approaches have been proposed~\cite{FRISonoff, FRISsec22, FRISsec, FRISambc, FRISlook}. %Nevertheless, many existing methods rely on restrictive assumptions or structural simplifications to render the problem tractable. In particular, some approaches are limited to uniform-linear-array (ULA) FRIS configurations~\cite{FRISsec}, while others simplify the problem structure by assuming a single-antenna BS~\cite{FRISonoff, FRISambc, FRISsec22} or by adopting idealized continuous phase-shift models that are not practically realizable under finite-resolution hardware constraints~\cite{FRISlook, FRISambc, FRISsec22}. Moreover, the direct BS-user link is often ignored~\cite{FRISonoff, FRISsec22, FRISambc, FRISlook}, or the resulting designs provide no guarantees on near-optimality~\cite{FRISlook, FRISsec22, FRISambc, FRISonoff, FRISsec}.
{Nevertheless, many existing approaches cast the FRIS configuration problem as a combinatorial optimization task and resort to either simplifying structural assumptions or computationally intensive relaxations to achieve tractability. While assumptions such as ULA-based FRIS configurations~\cite{FRISsec} or single-antenna BS models~\cite{FRISonoff, FRISambc, FRISsec22} can often be extended to more general settings, they do not fundamentally address the intrinsic difficulty of the problem. In particular, the key challenge lies in the strong coupling between port selection and discrete phase configuration, which leads to a highly nonconvex design space under practical finite-resolution constraints. As a result, existing methods typically resort to stochastic iteration techniques~\cite{FRISonoff, FRISsec} or relaxation-based~\cite{FRISlook} approaches with extremely high complexity of six to seventh order of scale, which often fail to reflect their inherent realistic joint structure and computability into the optimal configuration. Moreover, many existing designs either neglect the direct BS-user link~\cite{FRISonoff, FRISsec22, FRISambc, FRISlook} or do not provide guarantees on near-optimality~\cite{FRISlook, FRISsec22, FRISambc, FRISonoff, FRISsec}, further limiting their practical effectiveness.} Consequently, the fundamental geometry governing FRIS signal combining under multi-antenna systems and practical hardware constraints remains insufficiently characterized, leaving the underlying unified structure of the FRIS configuration problem largely unexplored.
Motivated by these challenges, this paper investigates the beamforming-gain maximization problem in an FRIS-assisted downlink system under multi-antenna BS and practical finite-resolution phase constraints, from a geometric perspective grounded in Minkowski-sum representations of FRIS codebook spaces superpositioned by channel coefficients~\cite{dohminkow, dohnps}. {In contrast to approaches that depend on specific modeling assumptions or treat the problem primarily as a combinatorial optimization task, the proposed framework aims to uncover and exploit the intrinsic geometric structure of the FRIS configuration. This geometric perspective is independent of particular array geometries or antenna configurations, thereby providing a more general and unified foundation for FRIS design. Moreover, by explicitly capturing the dominant combining direction through this structure, the proposed method achieves improved performance with more efficient and reliable configuration compared to existing approaches.} The main contributions of this work are summarized as follows:

{\begin{itemize}
\item \textbf{FRIS beamforming problem formulation and alternative optimization (AO) framework:} We formulate a beamforming-gain maximization problem for an FRIS-assisted downlink system with a multi-antenna BS and the direct BS-user link, where the beamformer, the selected ports, and the finite-resolution unit-modulus phases are jointly optimized. To handle the intrinsic coupling between continuous beamforming and discrete FRIS configuration, we develop an AO framework that alternates a closed-form maximum-ratio transmission (MRT) update with a globally optimal FRIS update, and prove that the objective is monotonically nondecreasing and becomes constant within finitely many iterations.
\item \textbf{Minkowski-geometry-based optimal FRIS configuration:} Extending the Minkowski-sum framework for discrete beamforming~\cite{dohminkow} to port selection, we represent the feasible reflected-sum set as a union of Minkowski sums over all $M_o$-port subsets and show via support-function identities that, for a given beamformer, the combinatorial port selection reduces to the Top-$M_o$ selection of direction-dependent port-wise scores. This transforms the FRIS subproblem into a one-dimensional maximization over a directional parameter, and the globally optimal port set and phase codewords are obtained by evaluating finitely many directions determined by the codeword- and Top-$M_o$-switching angles.
\item \textbf{Closed-form characterization for practical phase codebooks:} For regular $M_p$-gon phase-shifter codebooks, we derive closed-form port-wise scores and a nearest-grid codeword selection rule, and reveal a piecewise-sinusoidal structure of the resulting support function. Consequently, the global optimum is identified by comparing finitely many candidate resultants, generated by an angular sweep with near-quadratic complexity %$\mathcal O(M^2M_p\log M)$ 
per AO iteration, in contrast to the cubic scaling of existing stochastic approaches.
\item \textbf{Robustness analysis under imperfect channel-state information (CSI):} Under bounded CSI errors, we derive explicit bounds on the beamforming-gain mismatch and the optimality loss of the configuration obtained using imperfect CSI. We also establish a sufficient separation-margin condition under which the optimal selected-port set remains unchanged.
\item \textbf{Performance and near-optimality validations:} Extensive simulations demonstrate superior performance over benchmark schemes, near-optimal beamforming gains relative to exhaustive search, rapid AO convergence, and robustness to CSI uncertainty.
\end{itemize}}

\section{System Model}
\label{sec:system_model}
\begin{figure}[t]
  \begin{center}
    \includegraphics[width=0.55\columnwidth,keepaspectratio]{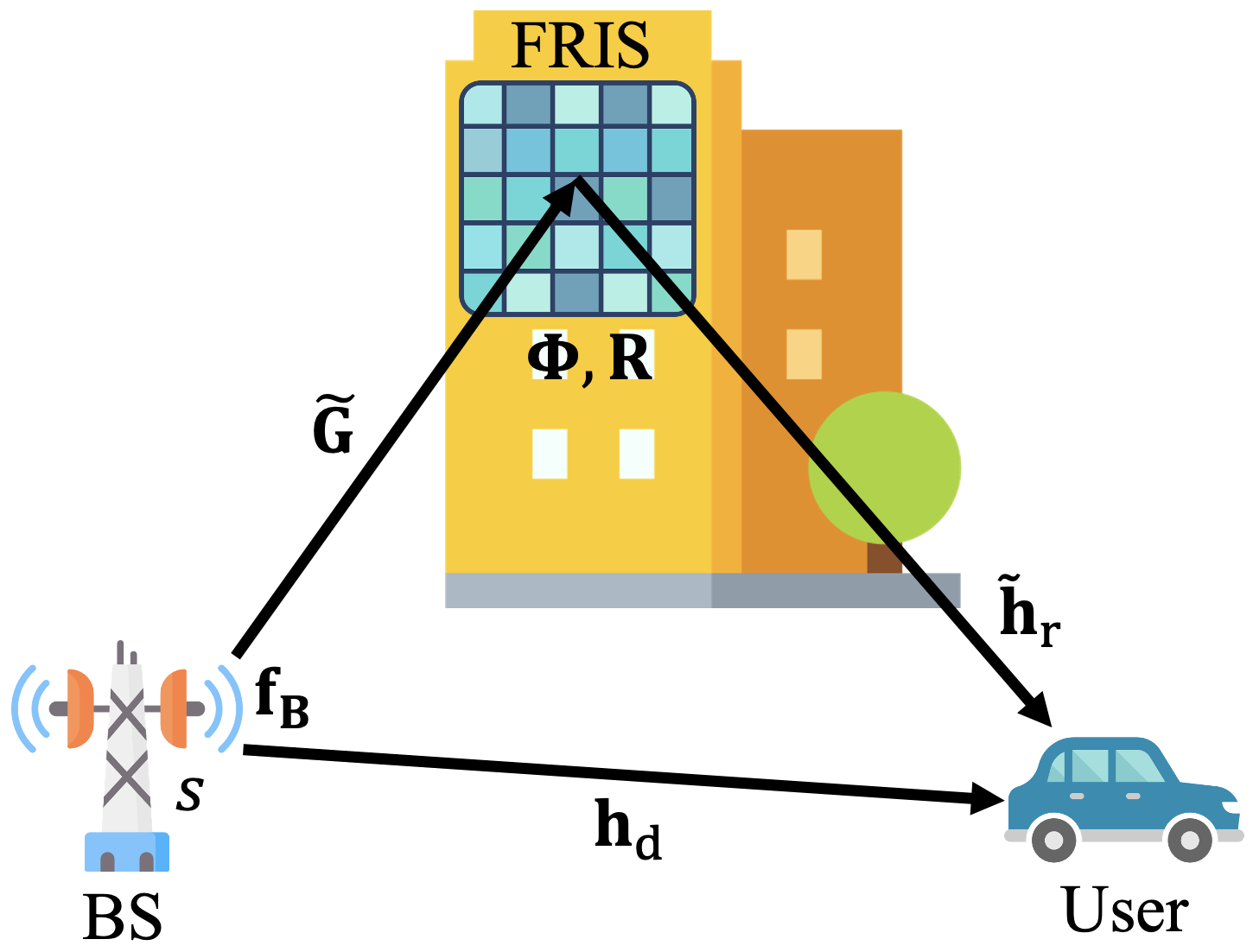}
    \caption{{FRIS-assisted downlink system.}}
    \label{fig_sys}
  \end{center}
\end{figure}
\subsection{Signal Model}
As depicted in Fig.~\ref{fig_sys}, we consider a downlink system in which an $N$-antenna BS serves a single-antenna user with the aid of an FRIS. The BS transmits a unit-power symbol $s$ using a beamformer $\mathbf f_{\mathrm B}\in\mathbb C^{N}$ under the power constraint: $\|\mathbf f_{\mathrm B}\|_2^2 \le P$. {The FRIS occupies a prescribed $M=M_x\times M_x$ elements uniformly arranged on a square surface of size $W_x\lambda \times W_x\lambda$, where $\lambda$ is the carrier wavelength and $W_x$ denotes the aperture normalized by $\lambda$. The corresponding inter-element spacing is ${d_{\mathrm s}}=\frac{W_x\lambda}{M_x-1}$. Herein, note that $M$ specifies the spatial resolution of the position-selection grid and does not represent the number of simultaneously deployed reflecting elements~\cite{FRISmag}. The $M_o(\le M)$ fluid reflecting elements can be selected herein by connecting them to the control circuit to modulate the incident signal, where the $M-M_o$ leftovers are connected to a matched load circuit and do not modulate the incident signal~\cite{FRISmag}. Accordingly, the position-domain configuration of FRIS is represented by $\mathbf w=[w_1 \cdots w_M]^{\mathrm T}$ subject to the $M_o$-selection constraint~\cite{FRISonoff,FRISsec}: $|\{m\in\{1,\cdots,M\}: w_m\neq 0\}|=M_o$. Herein, $w_m\neq 0$ indicates that the $m$th candidate position of FRIS is selected to modulate the signal; each $w_m$ is chosen from a finite unit-modulus codebook $ \mathcal W_m \subset \{w\in\mathbb C: |w|=1\}$, which models discrete phase resolution for implementation in practice. Since $w_m=0$ indicates that this candidate position is not viable, we define the augmented feasible set for FRIS port $m$: 
\begin{equation}
\label{eq:theta_def}
\Theta_m \triangleq \{0\}\cup\mathcal W_m~(m=1,\cdots,M),
\end{equation}
which models the fixed number ($M_o$) of available elements, rather than the deliberate deactivation of $M-M_o$ already-installed and independently controllable RIS elements.\footnote{{The constraint is therefore a hardware-count constraint and is unrelated to active RF amplification or the reflection power consumption of passive elements. If \(M\) independently controllable reflecting elements are available, one can set \(M_o=M\); this recovers the fully populated RIS case and eliminates the position-selection step.}}
}

%Accordingly, the FRIS configuration is described by $\mathbf w=[w_1 \cdots w_M]^{\mathrm T}$ subject to the $M_o$-selection constraint~\cite{FRISonoff,FRISsec}: $|\{m\in\{1,\cdots,M\}: w_m\neq 0\}|=M_o$, where $w_m\neq 0$ indicates that port $m$ is selected. {Herein, each $w_m$ is chosen from a finite unit-modulus codebook $ \mathcal W_m \subset \{w\in\mathbb C: |w|=1\}$, which models discrete phase resolution for implementation in practice. Since the zero-coefficient corresponds to unselected port, we define the augmented feasible set for FRIS port $m$: 
%\begin{equation}
%\label{eq:theta_def}
%\Theta_m \triangleq \{0\}\cup\mathcal W_m~(m=1,\cdots,M).
%\end{equation}}
Let $\mathbf p_m$ denote the position of the $m$th FRIS candidate port. Due to the close spacing between the elements, we model spatial correlation under isotropic scattering via the zeroth-order spherical Bessel function $j_0 (x)$~\cite{FRISpa, FRISonoff}. Thereby, we can define the FRIS correlation matrix $\mathbf R\succeq 0$ as  $[\mathbf R]_{m,i}=j_0\left(\frac{2\pi}{\lambda}\|\mathbf p_m-\mathbf p_i\|_2\right)$.

%Let $\widetilde{\mathbf G}\in\mathbb C^{M\times N}$ and $\widetilde{\mathbf h}_{\mathrm r}\in\mathbb C^{M}$ have Rician and Rayleigh fading under rich scattering conditions, respectively~\cite{FRISsec}.
{We model the BS-FRIS and the FRIS-user channel as $\mathbf G=\mathbf R^{1/2}\widetilde{\mathbf G},~\mathbf h_{\mathrm r}=\mathbf R^{1/2}\widetilde{\mathbf h}_{\mathrm r}$, respectively, where $\widetilde{\mathbf G}\in\mathbb C^{M\times N}$ and $\widetilde{\mathbf h}_{\mathrm r}\in\mathbb C^{M}$ are the uncorrelated channel between BS-FRIS and FRIS-user, respectively}, and the direct BS-user channel is denoted by $\mathbf h_{\mathrm d}\in\mathbb C^{N}$. %which also follows Rayleigh fading under rich scattering conditions~\cite{FRISsec}. 
We assume that the BS has access to perfect CSI during the optimization process~\cite{FRISmag}, while the impact of imperfect CSI is analyzed subsequently. Thereby, the received signal $y$ is expressed as $y=\Big(\mathbf h_{\mathrm d}^{*}\mathbf f_{\mathrm B}+\mathbf h_{\mathrm r}^{*}\mathbf \Phi\mathbf G\mathbf f_{\mathrm B}\Big)s+n=\mathbf a^* \mathbf f_{\mathrm{B}} s+n$, where $\mathbf a^*\triangleq\mathbf h_{\mathrm d}^*+\mathbf h_{\mathrm r}^{*}\mathbf \Phi\mathbf G$, $n\sim\mathcal{CN}(0,\sigma^2)$ is additive noise and $\mathbf \Phi=\mathrm{diag}(\mathbf w)\in\mathbb C^{M\times M}$ is the FRIS configuration matrix. Herein, $y$ can be equivalently written as
\begin{equation}
\label{eq:rx_scalar_dir}
y=\mathbf h_{\mathrm d}^{*}\mathbf f_{\mathrm B}s+\sum_{m=1}^{M} w_m h_{\mathrm r,m}^{*}\mathbf G_{m,:}\mathbf f_{\mathrm B} s+n,
\end{equation}
{where {$h_{\mathrm r,m}^{*}$ denotes the complex conjugate of the $m$th element of $\mathbf h_{\mathrm r}$, and $\mathbf G_{m,:}$ denotes the $m$th row of $\mathbf G$.}} For convenience, define the cascaded $h_m \triangleq h_{\mathrm r,m}^{*}\mathbf G_{m,:}\mathbf f_{\mathrm B}~(m=1,\cdots,M)$, and define the direct-path $d\triangleq \mathbf h_{\mathrm d}^{*}\mathbf f_{\mathrm B}$. Then~\eqref{eq:rx_scalar_dir} yields the compact scalar form:%\footnote{More advanced mutual coupling-aware and physics-consistent RIS channel models, such as scattering-matrix-based formulations and multiport-network representations~\cite{r1r, r2r}, may exhibit nonlinear electromagnetic interactions at the element level. Nevertheless, their effects are ultimately embedded in the effective channel responses and can be equivalently represented through $\{h_m\}$. As a result, the overall input-output relation still admits the same linear superposition structure.}
\begin{equation}
\label{eq:rx_signal_dir}
y=\Big(d+\sum_{m=1}^{M} w_m h_m\Big)s+n.
\end{equation}

\subsection{Problem Formulation}
\label{sec:problem_formulation}
From~\eqref{eq:rx_signal_dir}, we define the effective channel coefficient $z$ as $z \triangleq d+\sum_{m=1}^{M} w_m h_m$, whose magnitude represents the resulting beamforming gain at the user. Accordingly, the maximization of beamforming gain is formulated as
\begin{equation}
\label{prob:fris_exact_dir}
\begin{aligned}
&\max_{\{w_m\}, \Gamma, \mathbf f_{\mathrm{B}}}~ \left| d+\sum_{m=1}^{M} w_m h_m \right| \\
\text{s.t.}~ 
& w_m \in \Theta_m~(m=1,\cdots,M),~\|\mathbf f_{\mathrm B}\|_2^2\le P,\\
& \big|\underbrace{\{m\in\{1,\cdots,M\}: w_m\neq 0\}}_{\triangleq\Gamma}\big| = M_o,
\end{aligned}
\end{equation}
where $\Gamma\subseteq\{1,\cdots,M\}$ {denotes} the selected-port index set with $|\Gamma|=M_o$. For a given $\Gamma$, set $w_m=0$ for all $m\notin\Gamma$ and optimize $\{w_m\in\mathcal W_m\}_{m\in\Gamma}$. Then~\eqref{prob:fris_exact_dir} is equivalent to% choosing $\Gamma$ and $\{w_m\}_{m\in\Gamma}$ that maximize
\begin{equation}
\label{eq:z_Gamma_dir}
\begin{aligned}
&\max_{\{w_m\}, \Gamma, \mathbf f_{\mathrm{B}}}~\left|d+\sum_{m\in\Gamma} h_m w_m\right|\\
&\text{s.t.}~w_m\in\mathcal W_m~(m\in\Gamma), \|\mathbf f_{\mathrm B}\|_2^2\le P{,}\\
&\quad~~{\Gamma\subseteq\{1,\cdots,M\},~|\Gamma|=M_o.}
\end{aligned}
\end{equation}
Due to the strong coupling between $\mathbf f_{\mathrm B}$, $\{w_m\}$, and $\Gamma$ through the effective channel gain, together with the combinatorial nature of choosing $\Gamma$ from $\{1,\cdots,M\}$,~\eqref{eq:z_Gamma_dir} becomes highly nonconvex, making direct joint optimization intractable. To address this difficulty, we employ an AO framework that decomposes the original problem into two subproblems; beamformer design and FRIS configuration, which are solved in an iterative manner.

\section{Proposed Framework}
\subsection{Update of $\mathbf f_{\mathrm{B}}$}
\label{subsec:bf_closed_form}
For fixed $\{w_m\}$ and $\Gamma$,~\eqref{eq:z_Gamma_dir} reduces to
\begin{equation}
\label{prob:bf_sub}
\max_{\mathbf f_{\mathrm B}}~\big|\mathbf a^{*}\mathbf f_{\mathrm B}\big|~\text{s.t.}~\|\mathbf f_{\mathrm B}\|_2^2\le P,
\end{equation}
whose solution $\mathbf f_{\mathrm B}^\star$ is given by MRT: $\mathbf f_{\mathrm B}^\star=\sqrt{P}\frac{\mathbf a}{\|\mathbf a\|_2}$.
\subsection{Proposed Approach: Update of $\{w_m\}$ and $\Gamma$ (General Case)}
\label{subsec:rigorous_solution_phi}
\subsubsection{Transformation to FRIS-Only Subproblem}
For fixed $\mathbf f_{\mathrm{B}}$, plugging it into~\eqref{prob:fris_exact_dir}, the original problem is equivalently reduced to
\begin{equation}
\label{prob:fris_reduced22}
\begin{aligned}
&\max_{\{w_m\}, \Gamma}~ \left| d+\sum_{m=1}^{M} w_m h_m \right| \\
\text{s.t.}~ 
& w_m \in \Theta_m~(m=1,\cdots,M),\\
& \big|\underbrace{\{m\in\{1,\cdots,M\}: w_m\neq 0\}}_{\triangleq\Gamma}\big| = M_o.
\end{aligned}
\end{equation}
Define the feasible reflected-sum set with respect to $\Gamma$ as $\mathcal Z(\Gamma)
\triangleq
\sum_{m\in\Gamma} h_m \mathcal W_m$, which is a Minkowski sum of finite sets. The overall feasible reflected-sum set under $M_o$-port selection is the finite union
\begin{equation}
\label{eq:Z_Mo}
\mathcal Z_{M_o}
\triangleq
\bigcup_{\substack{\Gamma\subseteq\{1,\cdots,M\}\\|\Gamma|=M_o}}
\mathcal Z(\Gamma).
\end{equation}
With $d$, the overall feasible combined-sum set is the translation: $\mathcal Z_{M_o,t} \triangleq d+\mathcal Z_{M_o}=\{d+z:  z\in\mathcal Z_{M_o}\}$. Therefore,~\eqref{prob:fris_reduced22} can be equivalently rewritten as
\begin{equation}
\label{prob:minkowski_union_dir}
\max_{z}~ |z|~\text{s.t.}~ z\in \mathcal Z_{M_o,t}.
\end{equation}
\subsubsection{Equivalent Transformations}
The proposed method proceeds through the following equivalent transformations:
\begin{equation}
\label{eq:equiv_chain_dir}
\begin{aligned}
\max_{z\in \mathcal Z_{M_o,t}}|z|
\overset{(a)}{=}
\max_{z\in\mathcal P_{M_o,t}}|z|
&\overset{(b)}{=}
\max_{|u|=1} h_{\mathcal P_{M_o,t}}(u)\\
&\overset{(c)}{=}
\max_{\phi\in[0,2\pi)} h_{\mathcal P_{M_o,t}}(e^{j\phi}),
\end{aligned}
\end{equation}
where {$h_{\mathcal C}(u)\triangleq\max_{z\in\mathcal C}\Re\{u^*z\}~(|u|=1)$ denotes the support function of a nonempty compact set $\mathcal C\subset\mathbb C$, which satisfies $h_{\mathcal C}=h_{\Conv(\mathcal C)}$~\cite{cbmt}}, and $\mathcal P_{M_o,t}\triangleq\Conv(\mathcal Z_{M_o,t})=\Conv(d+\mathcal Z_{M_o})=d+\mathcal P_{M_o}$, where $\Conv(\cdot)$ is the convex hull of the given set~\cite{boyd}, and $\mathcal P_{M_o}\triangleq\Conv(\mathcal Z_{M_o})$~\cite{cbmt}. This reduces the original mixed discrete optimization into a one-dimensional search over $\phi$. {Herein, equality~(a) holds by~\cite[Lemma~1]{dohminkow}, and (b)-(c) holds by~\cite[Equation~(10)]{dohminkow}{, i.e., by the identity $|z|=\max_{|u|=1}\Re\{u^*z\}$ followed by the interchange of the two maximizations}. Hence, it remains to evaluate $h_{\mathcal P_{M_o,t}}(e^{j\phi})$ efficiently for given $\phi$. %, by identity $|z|=\max_{|u|=1}\Re\{u^*z\}$ with maximizer $u^\star=e^{j\arg(z)}$ when $z\neq 0$, }
The geometry in Fig.~\ref{fig_geo} illustrates the interpretation of $h_{\mathcal C}(u)$, %For a general $\mathcal C\subset\mathbb C$ in Fig.~\ref{fig_geo1}, $h_{\mathcal C}(u)$ represents the maximum projection of $\mathcal C$ onto the direction $u=e^{j\phi}$, obtained by translating a line orthogonal to $u$ until it first touches the boundary of $\mathcal C$. As $\phi$ varies, the supporting point moves continuously along the boundary. By contrast, when $\mathcal C$ is a polygon induced by a finite phase codebook in Fig.~\ref{fig_geo2}, the maximizer is always attained at a vertex~\cite{boyd, cbmt}. Consequently, $h_{\mathcal C}(u)$ becomes piecewise linear in $\phi$, with the maximizing vertex remaining unchanged over angular intervals and switching at the bisectors between adjacent vertices, leading to a direction-dependent vertex selection.
which enables an algorithmic realization for general and regular $M_p$-gon codebooks. Therein, the optimal vertex, interpretable as the one whose associated Thales circle dominates the projection along a given direction~\cite{dohminkow}, and hence the optimal quantized phase, can be efficiently identified through direction-dependent interval partitioning.}
%\begin{figure}[t]
%  \begin{center}
%    \includegraphics[width=0.5\columnwidth,keepaspectratio]{Fig/support}
%    \caption{Geometrical illustration of $h_{\mathcal C}(u)$ with respect to general and polygon type of $\mathcal C$.}
%    \label{fig_geo}
%  \end{center}
%\end{figure}
\begin{figure}[t]
    \centering
    \subfloat[]{%
        \includegraphics[width=0.24\textwidth]{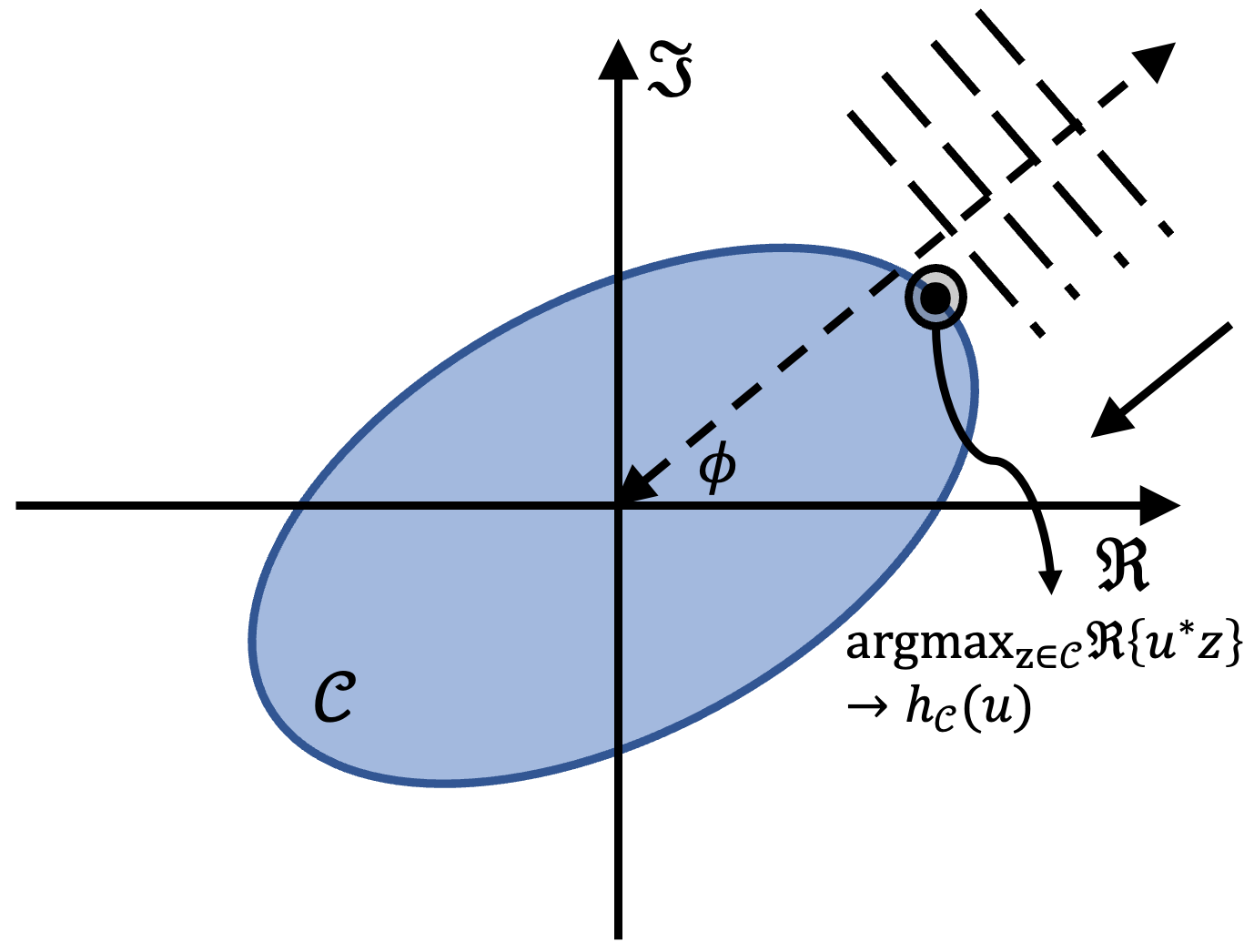}
        \label{fig_geo1}%
    }
    \subfloat[]{%
        \includegraphics[width=0.24\textwidth]{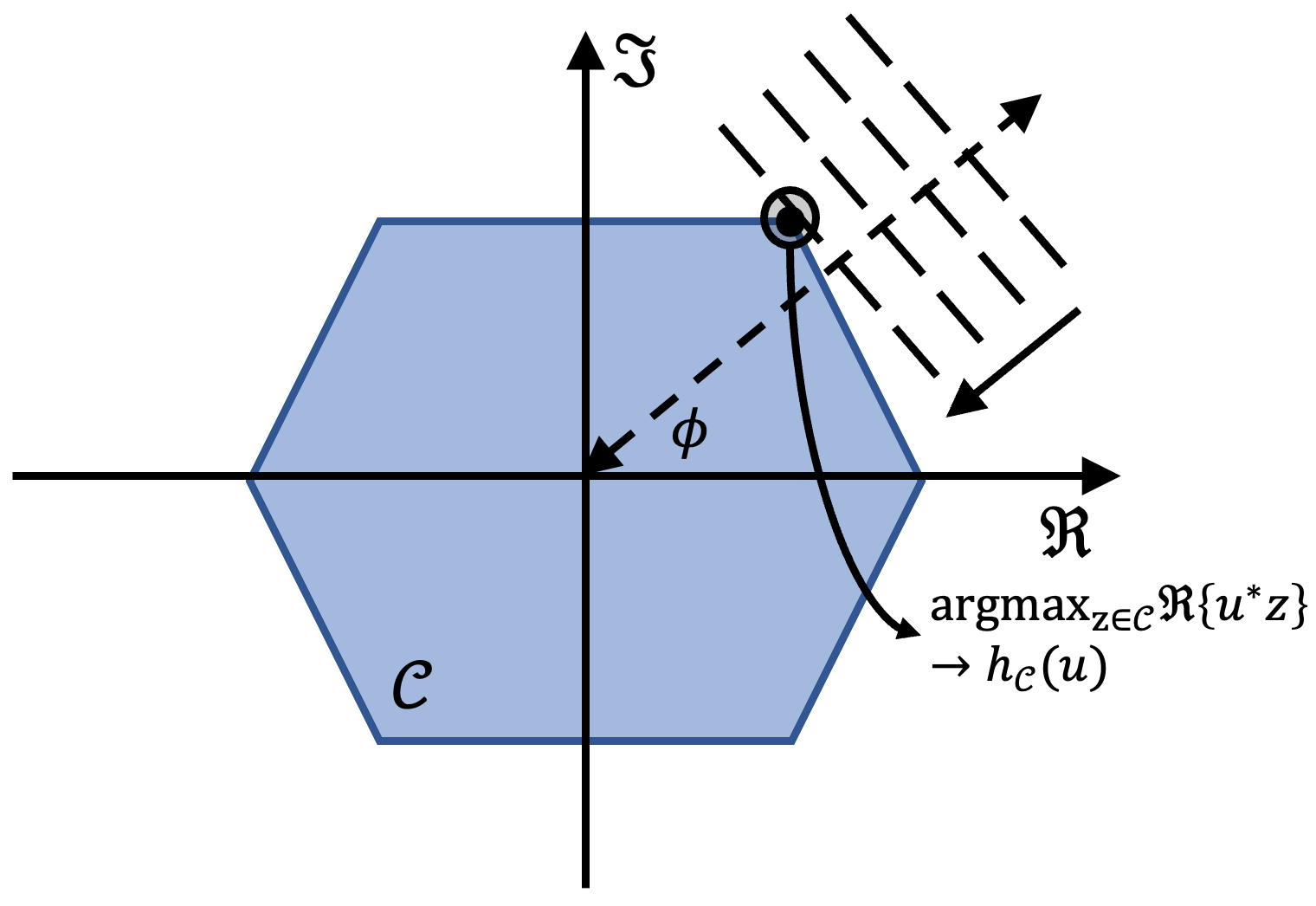}
        \label{fig_geo2}%
    }
    \caption{Geometrical illustration of $h_{\mathcal C}(u)$ with respect to (a) general and (b) polygon type of $\mathcal C$.}
    \label{fig_geo}
\end{figure}

\subsubsection{Evaluation of $h_{\mathcal P_{M_o,t}}(e^{j\phi})$}
\paragraph{Support-Function Identities and Translation Decomposition}
{We use the following standard identities of support functions~\cite{cbmt}:
\begin{equation}
\label{eq:support_identities_dir}
%\begin{aligned}
%&h_{\Conv(\mathcal S)}(u)=\max_{z\in\mathcal S}\Re\{u^*z\},\\
h_{\Conv(\cup_i \mathcal C_i)}(u)=\max_i h_{\mathcal C_i}(u),~h_{\mathcal C_1+\mathcal C_2}(u)=h_{\mathcal C_1}(u)+h_{\mathcal C_2}(u).
%\end{aligned}
\end{equation}
By definition of $\mathcal{P}_{M_o,t}$, the support function decomposes as
\begin{equation}
\label{eq:support_translation_dir}
h_{\mathcal P_{M_o,t}}(e^{j\phi})
=
h_{\{d\}}(e^{j\phi})+h_{\mathcal P_{M_o}}(e^{j\phi})
=
\Re\{e^{-j\phi}d\}+h_{\mathcal P_{M_o}}(e^{j\phi}).
\end{equation}
Note that the direct link contributes an additive directional term that depends only on $\phi$ and does not couple across ports.}
\paragraph{Support Function of $\mathcal P_{M_o,t}$ and Portwise Scores}
{Applying}~\eqref{eq:support_identities_dir} to $\mathcal P_{M_o}=\Conv(\mathcal Z_{M_o})$ yields, for any $|u|=1$,
\begin{equation}
\label{eq:support_PMo_geom_only}
h_{\mathcal P_{M_o}}(u)=\max_{\substack{\Gamma\subseteq\{1,\cdots,M\}\\|\Gamma|=M_o}}\sum_{m\in\Gamma} h_{h_m\Conv(\mathcal W_m)}(u).
\end{equation}
Let $u=e^{j\phi}$ with $\phi\in[0,2\pi)$. Define the port-wise score
\begin{equation}
\label{eq:score_geom_only_repeat}
\begin{aligned}
s_m(\phi)
\triangleq~
&h_{h_m\Conv(\mathcal W_m)}(e^{j\phi})\\
=~
&\max_{w\in\mathcal W_m}\Re\{e^{-j\phi}h_m w\}~ (m=1,\cdots,M),
\end{aligned}
\end{equation}
which leads to
\begin{equation}
\label{eq:support_topMo_geom_only_repeat}
h_{\mathcal P_{M_o}}(e^{j\phi})
=
\max_{\substack{\Gamma\subseteq\{1,\cdots,M\}\\|\Gamma|=M_o}}
\sum_{m\in\Gamma}s_m(\phi).
\end{equation}
Combining~\eqref{eq:support_translation_dir} and~\eqref{eq:support_topMo_geom_only_repeat} yields
\begin{equation}
\label{eq:support_PMo_dir_final}
h_{\mathcal P_{M_o,t}}(e^{j\phi})
=
\Re\{e^{-j\phi}d\}
+
\max_{\substack{\Gamma\subseteq\{1,\cdots,M\}\\|\Gamma|=M_o}}
\sum_{m\in\Gamma}s_m(\phi).
\end{equation}
\subsubsection{Top-$M_o$ Induced FRIS Port Selection}
\label{tifp}
Since the second term in~\eqref{eq:support_PMo_dir_final} is a linear sum of per-port scores under a cardinality constraint, an optimizer is obtained by selecting the $M_o$ largest scores:
\begin{equation}
\label{eq:Gamma_phi_geom_only}
\Gamma^\star(\phi)
\in
\argmax_{\substack{\Gamma\subseteq\{1,\cdots,M\}\\|\Gamma|=M_o}}
\sum_{m\in\Gamma}s_m(\phi)
=
\mathrm{Top}\text{-}M_o\big(\{s_m(\phi)\}_{m=1}^M\big),
\end{equation}
where ties may be broken arbitrarily. Note that $\Re\{e^{-j\phi}d\}$ does not affect this selection rule since it is independent of $\Gamma$. Given $\Gamma^\star(\phi)$, for each selected port choose
\begin{equation}
\label{eq:w_hat_phi_geom_only}
\hat w_m(\phi)\in\argmax_{w\in\mathcal W_m}\Re\{e^{-j\phi}h_m w\}~(m\in\Gamma^\star(\phi)),
\end{equation}
and set $w_m=0$ for $m\notin\Gamma^\star(\phi)$. This construction attains $h_{\mathcal P_{M_o}}(e^{j\phi})$ in~\eqref{eq:support_topMo_geom_only_repeat}.
\begin{figure}[t]
  \begin{center}
    \includegraphics[width=0.98\columnwidth,keepaspectratio]{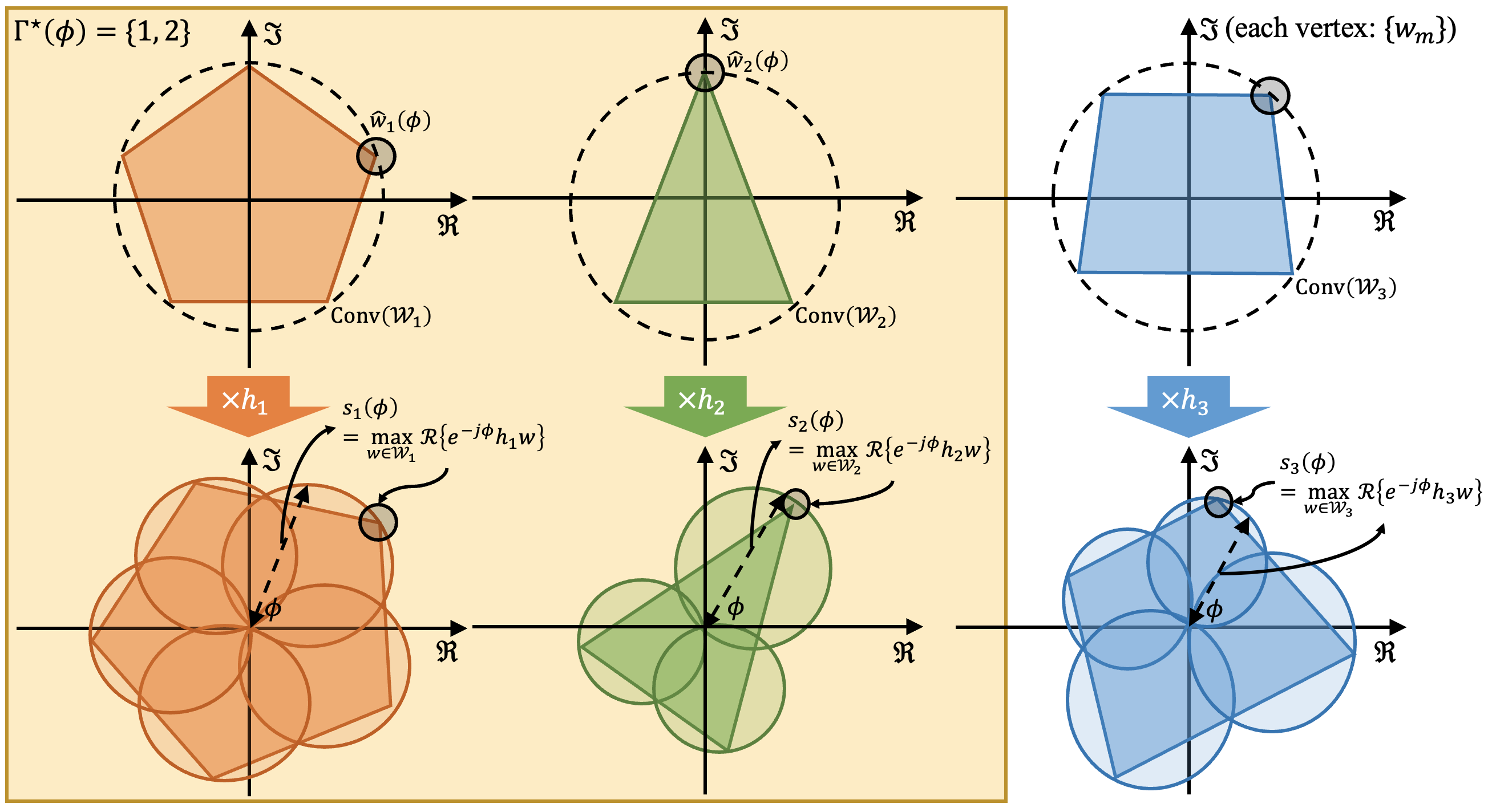}
    \caption{Illustration and the resulting Top-$M_o$ port selection based on arc dominance in the Thales-circle representation with $M=3$ and $M_o=2$.}
    \label{fig_geo32}
  \end{center}
\end{figure}

Herein, we provide a geometric interpretation of the proposed FRIS port selection algorithm, illustrated by Fig.~\ref{fig_geo32}, by viewing $s_m(\phi)$ as a directional projection over $h_m\Conv(\mathcal W_m)$. Herein, each $h_m w~(w\in\mathcal W_m)$ induces a Thales circle $e^{j\theta}\Re\{e^{-j\theta}h_m w\}~(\theta\in[0, 2\pi))$ with diameter $[0,h_m w]$, whose projection along $e^{j\phi}$ for a given $\phi$ is $\Re\{e^{-j\phi}h_m w\}$~\cite{dohminkow}. As $\phi$ varies, each vertex dominates a specific angular interval corresponding to an outer arc of its associated circle, as depicted in the lower part of Fig.~\ref{fig_geo32}. Accordingly, $s_m(\phi)$ is determined by the vertex whose arc contains $\phi$, i.e., the one that yields the largest projection in the direction $e^{j\phi}$, which coincides with the interpretation in Fig.~\ref{fig_geo2}. This geometric optimization is carried out independently for each $m$. The subsequent Top-$M_o$ operation, which yields $\Gamma^\star(\phi)$, then selects the $M_o$ ports with the largest individually optimized $s_m(\phi)$.

Therefore, the optimal FRIS configuration is obtained by the one-dimensional search with respect to $\phi$:
%\begin{equation}
%\label{eq:phi_star_geom_only_support_dir}
%\phi^\star \in \argmax_{\phi\in[0,2\pi)}h_{\mathcal P_{M_o,t}}(e^{j\phi})
%%\left[
%%\Re\{e^{-j\phi}d\}
%%+
%%\sum_{m\in\Gamma^\star(\phi)} s_m(\phi)
%%\right],
%\end{equation}
\begin{equation}
\label{eq:phi_star_geom_only_support_dir}
\phi^\star \in \argmax_{\phi\in[0,2\pi)}\left(\underbrace{\Re\{e^{-j\phi}d\}+\sum_{m\in\Gamma^\star(\phi)} s_m(\phi)}_{=h_{\mathcal P_{M_o,t}}(e^{j\phi})}\right)\\
\end{equation}
and the corresponding final FRIS coefficients are
\begin{equation}
\label{eq:w_star_geom_only}
w_m^\star=
\begin{cases}
\hat w_m(\phi^\star) & (m\in\Gamma^\star(\phi^\star)),\\
0 & (\text{otherwise}).
\end{cases}
\end{equation}
The achieved combined sum is
\begin{equation}
\label{eq:z_star_dir}
z^\star
=
d+\sum_{m\in\Gamma^\star(\phi^\star)} h_m \hat w_m(\phi^\star),
\end{equation}
and $|z^\star|$ equals the optimal value of~\eqref{prob:fris_reduced22}. The overall FRIS configuration procedure is summarized in Algorithm~\ref{alg:geom_support_search_dir}.
{
\begin{remark}
\label{rem:general_partition}
For a general finite unit-modulus codebook $\mathcal W_m=\{e^{j\theta_{m,1}},\cdots,e^{j\theta_{m,K_m}}\}$ with $0\le\theta_{m,1}<\cdots<\theta_{m,K_m}<2\pi$ and $\theta_{m,K_m+1}\triangleq\theta_{m,1}+2\pi$, the maximizer in~\eqref{eq:w_hat_phi_geom_only} is the codeword whose phase is circularly closest to $\phi-\alpha_m$, where $\alpha_m\triangleq\arg(h_m)$. Hence, the codeword of port $m$ switches only at $\alpha_m+\frac{\theta_{m,i}+\theta_{m,i+1}}{2}~(\mathrm{mod}~2\pi)$ $(i=1,\cdots,K_m)$, and the Top-$M_o$ switching angles follow from~\eqref{eq:score_tie_linear} in Section~\ref{subsec:top_switching} with $\beta_m=\alpha_m+\theta_{m,i}$. On each interval of the resulting partition, the configuration $(\Gamma^\star(\phi),\{\hat w_m(\phi)\})$ is constant, so that the argument of Corollary~\ref{cor:max_resultant} applies verbatim. Note that the selection in Algorithm~\ref{alg:geom_support_search_dir} is based on $|z(\bar\phi)|$ rather than on $h_{\mathcal P_{M_o,t}}(e^{j\bar\phi})$: since $h_{\mathcal P_{M_o,t}}(e^{j\bar\phi})=\Re\{e^{-j\bar\phi}z(\bar\phi)\}\le|z(\bar\phi)|$, comparing $h_{\mathcal P_{M_o,t}}(e^{j\bar\phi})$ at arbitrary representatives can miss the optimal configuration, whereas $\max_{\bar\phi}|z(\bar\phi)|=|z^\star|$.
\end{remark}
}
\begin{algorithm}[t]
\caption{Proposed FRIS Configuration Algorithm}
\label{alg:geom_support_search_dir}
\begin{algorithmic}[1]
\Require $\{h_m\}_{m=1}^M$, $\{\mathcal W_m\}_{m=1}^M$, $M_o$, $d$
\State \multiline{Partition $[0,2\pi)$ by the codeword-switching angles of all ports and the Top-$M_o$ switching angles as in Remark~\ref{rem:general_partition}, and pick one representative $\bar\phi$ per interval.}
\For{each representative $\bar\phi$}
    \For{$m=1$ to $M$}
        \State \multiline{Compute $s_m(\bar{\phi})$ using~\eqref{eq:score_geom_only_repeat} (or~\eqref{eq:score_closed_form} for the regular $M_p$-gon case) by determining $w\in\mathcal W_m$ according to the procedure illustrated in Fig.~\ref{fig_geo32}.}
    \EndFor
    \State Determine $\Gamma^\star(\bar\phi)$ via~\eqref{eq:Gamma_phi_geom_only}.
    \State Evaluate $\hat w_m(\bar\phi)~(m\in\Gamma^\star(\bar\phi))$ by~\eqref{eq:w_hat_phi_geom_only}.
    \State {Evaluate $z(\bar\phi)\triangleq d+\sum_{m\in\Gamma^\star(\bar\phi)}h_m\hat w_m(\bar\phi)$.}
\EndFor
\State \multiline{Choose $\bar\phi^\star\in\argmax_{\bar\phi}|z(\bar\phi)|$ and set $\phi^\star=\arg z(\bar\phi^\star)$.}
\State \multiline{Set $w_m^\star$ by~\eqref{eq:w_star_geom_only} {with $\phi^\star$ replaced by $\bar\phi^\star$, and $z^\star=z(\bar\phi^\star)$}.}
\State \Return $\Gamma^\star({\bar\phi^\star})$, $\{w_m^\star\}$, and $z^\star$
\end{algorithmic}
\end{algorithm}
\section{Update of $\{w_m\}$ and $\Gamma$ (Regular $M_{p}$-gon Codebook)}
\label{subsec:regular_Mgon}
In this section, we specialize the proposed framework in Section~\ref{subsec:rigorous_solution_phi} to the practically important case where each feasible reflection set is a uniform $M_{p}$-ary phase-shifter codebook~\cite{FRISonoff, FRISsec}, namely the vertices of a regular $M_{p}$-gon on the unit circle. This specialization yields (i) closed-form expressions for $s_m(\phi)$ and the corresponding optimal codeword, and (ii) a finite set of critical angles that avoids exhaustive angular sweeping.
\subsection{Regular $M_{p}$-gon Codebook and Quantization Residual}
\label{subsubsec:Mgon_codebook}
Assume that all ports share the same $M_{p}$-ary phase-shifter codebook $\mathcal W_m = \mathcal W_{M_{p}}
\triangleq
\Big\{ e^{j\frac{2\pi k}{M_{p}}} \Big\}_{k=0}^{M_{p}-1}~(\forall m)$, where $M_{p}\ge 2$ denotes the phase resolution. We then use the principal-value mapping $\mathrm{wrap}: \mathbb R\to(-\pi,\pi]$ defined by $\mathrm{wrap}(\phi)
\triangleq
{\phi - 2\pi \left\lceil \frac{\phi-\pi}{2\pi} \right\rceil}$, which maps $\phi$ to its unique equivalent in the principal interval $(-\pi,\pi]$, and define the quantization residual
\begin{equation}
\label{eq:delta_def}
\delta_{M_{p}}(\phi)
\triangleq
\min_{k\in\{0,\cdots,M_{p}-1\}}
\Big|\mathrm{wrap}\Big(\phi-\frac{2\pi k}{M_{p}}\Big)\Big|
\in \Big[0,\frac{\pi}{M_{p}}\Big], 
\end{equation}
which quantifies the minimum angular mismatch between $\phi$ and the nearest phase point in the $M_p$-ary quantized codebook.

Thereafter, the support function of the $M_{p}$-gon can be characterized by the following theorem:
\begin{theorem}
\label{lem:support_Mgon}
For any $\phi\in\mathbb R$,
\begin{equation}
\label{eq:support_Mgon}
h_{\Conv(\mathcal W_{M_{p}})}(e^{j\phi})=\max_{w\in\mathcal W_{M_{p}}}\Re\{e^{-j\phi}w\}=\cos\big(\delta_{M_{p}}(\phi)\big).
\end{equation}
Moreover, a maximizer is given by any
\begin{equation}
\label{eq:k_star_quant}
\begin{aligned}
&{\hat w(\phi)}=e^{j\frac{2\pi}{M_{p}}{k^\star(\phi)}},\\
&{k^\star(\phi)}\in\argmin_{k\in\{0,\cdots,M_{p}-1\}}\Big|\mathrm{wrap}\Big(\phi-\frac{2\pi k}{M_{p}}\Big)\Big|.\\
\end{aligned}
\end{equation}
\end{theorem}
\begin{proof}
Since $\Conv(\mathcal W_{M_{p}})$ is the convex hull of finitely many points, its support function is attained at a vertex:
\begin{equation}
\begin{aligned}
h_{\Conv(\mathcal W_{M_{p}})}(e^{j\phi})
&=\max_{k}\Re\{e^{-j\phi}e^{j\frac{2\pi k}{M_{p}}}\}\\
&=\max_{k}\cos\Big(\phi-\frac{2\pi k}{M_{p}}\Big).
\end{aligned}
\end{equation}
The maximizing $k$ is the nearest phase grid point (in principal value), which yields $\delta_{M_{p}}(\phi)$ in~\eqref{eq:delta_def}, proving~\eqref{eq:support_Mgon}. The maximizer characterization follows directly, giving~\eqref{eq:k_star_quant}.
\end{proof}

\subsection{Port and Direct-Term Score, and Optimal Codeword}
\label{subsubsec:score_closed_form}
Let $\alpha_m\triangleq \arg(h_m)$ and $a_m\triangleq |h_m|$ for all $m$ and let $\alpha_d\triangleq \arg(d)$ and $a_d\triangleq |d|$. Then $s_m(\phi)$ and the directional term admit the following closed-form:
\begin{corollary}
\label{lem:score_closed_form}
For any $\phi\in[0,2\pi)$,
\begin{equation}
\label{eq:score_closed_form}
s_m(\phi)=a_m\cos\Big(\delta_{M_{p}}(\phi-\alpha_m)\Big)~(m=1,\cdots,M).
\end{equation}
Moreover,
\begin{equation}
\label{eq:direct_term_closed_form}
\Re\{e^{-j\phi}d\}=a_d\cos(\phi-\alpha_d).
\end{equation}
An optimal FRIS codeword can be chosen as
\begin{equation}
\label{eq:w_hat_closed_form}
\begin{aligned}
&{\hat w_m(\phi)}=e^{j\frac{2\pi}{M_{p}}{k_m^\star(\phi)}},\\
&{k_m^\star(\phi)\triangleq k^\star(\phi-\alpha_m)}\in
\argmin_{k}
\Big|\mathrm{wrap}\Big(\phi-\alpha_m-\frac{2\pi k}{M_{p}}\Big)\Big|.\\
\end{aligned}
\end{equation}
\end{corollary}
\begin{proof}
By definition,
\begin{equation}
\label{bdsmp}
s_m(\phi)
=\max_{w\in\mathcal W_{M_{p}}}\Re\{e^{-j\phi}h_m w\}
%&=\max_{w\in\mathcal W_{M_{p}}}\Re\{e^{-j(\phi-\alpha_m)}a_m w\}\\
=a_m\max_{w\in\mathcal W_{M_{p}}}\Re\{e^{-j(\phi-\alpha_m)}w\}.
\end{equation}
Applying Theorem~\ref{lem:support_Mgon} with $\phi\leftarrow\phi-\alpha_m$ yields~\eqref{eq:score_closed_form}, and the maximizer follows from~\eqref{eq:k_star_quant} {with $\phi\leftarrow\phi-\alpha_m$}, proving~\eqref{eq:w_hat_closed_form}. The direct-term identity follows from $d=a_d e^{j\alpha_d}$ and $\Re\{e^{-j\phi}d\}=a_d\cos(\phi-\alpha_d)$, proving~\eqref{eq:direct_term_closed_form}.
\end{proof}
\begin{figure}[t]
  \begin{center}
    \includegraphics[width=0.98\columnwidth,keepaspectratio]{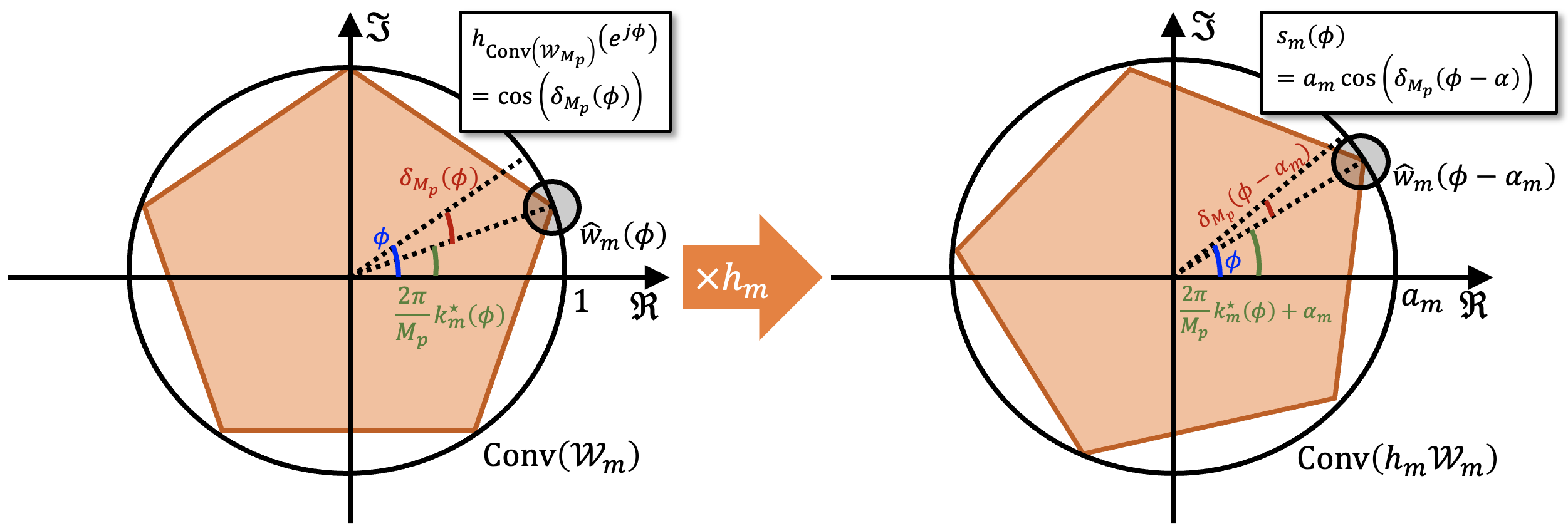}
    \caption{Geometric interpretation of $h_{\mathrm{Conv}(\mathcal W_{M_p})}(e^{j\phi})$ and $s_m(\phi)$ for a regular $M_p$-gon phase codebook and the resulting per-port score with $M_p=5$. The left figure depicts $\mathrm{Conv}(\mathcal W_{M_p})$ of the unit-modulus $M_p$-ary phase-shifter codebook and the evaluation of $h_{\mathrm{Conv}(\mathcal W_{M_p})}(e^{j\phi})=\cos(\delta_{M_p}(\phi))$ along direction $e^{j\phi}$, while the right figure illustrates the scaled and rotated convex set $\mathrm{Conv}(h_m\mathcal W_{M_p})$, where multiplication by $h_m=a_m e^{j\alpha_m}$ induces a rotation by $\alpha_m$ and a radial scaling by $a_m$, resulting in $s_m(\phi)=a_m\cos(\delta_{M_p}(\phi-\alpha_m))$.}
    \label{fig_mgonsm}
  \end{center}
\end{figure}
The geometrical illustration of deriving $h_{\mathrm{Conv}(\mathcal W_{M_p})}(e^{j\phi})$ and $s_m(\phi)$ is given in Fig.~\ref{fig_mgonsm}. Herein, the support function is interpreted as the maximum projection of the convex hull of the discrete phase codebook onto the direction $e^{j\phi}$, which is attained at the codeword whose phase is closest to $\phi$, leading to the cosine form governed by $\delta_{M_p}(\cdot)$. When incorporating $h_m$, this convex set undergoes a rotation by $\alpha_m$ and a scaling by $a_m$, so that the resulting $s_m(\phi)$ naturally captures both channel magnitude and phase quantization effects in a unified geometric manner.

With Corollary~\ref{lem:score_closed_form}, $s_m(\phi)$ and ${\hat w_m(\phi)}$ are available for $\phi\in[0,2\pi)$, and the objective in~\eqref{eq:phi_star_geom_only_support_dir} becomes
\begin{equation}
\label{eq:phi_obj_closed_form_dir}
h_{\mathcal P_{M_o,t}}(e^{j\phi})=a_d\cos(\phi-\alpha_d)+\sum_{m\in\Gamma^\star(\phi)}
a_m\cos\Big(\delta_{M_{p}}(\phi-\alpha_m)\Big),
\end{equation}
where $\Gamma^\star(\phi)=\mathrm{Top}\text{-}M_o(\{s_m(\phi)\})$.
Thus, for each $\phi$, we can evaluate $h_{\mathcal P_{M_o,t}}(e^{j\phi})$ with $\Gamma^\star(\phi)$ and ${\hat w_m(\phi)}~(m\in\Gamma^\star(\phi))$ by~\eqref{eq:w_hat_closed_form}, and assigning $\phi^\star$ and $w_m^\star$ as in~\eqref{eq:phi_star_geom_only_support_dir} and~\eqref{eq:w_star_geom_only}, respectively.

\subsection{Maximization of $h_{\mathcal P_{M_o,t}}(e^{j\phi})$ via a Finite Candidate Set}
\label{sec:finite_max_phi}
For evaluation of $h_{\mathcal P_{M_o,t}}(e^{j\phi})$ in~\eqref{eq:phi_obj_closed_form_dir}, although it is a sum of cosine-type terms, the dependence of (i) ${k_m^\star(\phi)}$ embedded in $\delta_{M_p}$ and (ii) $\Gamma^\star(\phi)$ on $\phi$ makes $h_{\mathcal P_{M_o,t}}(e^{j\phi})$ {piecewise smooth} with finitely many nondifferentiable points. In this subsection, we show that the global maximizer can be obtained by checking only a {finite} set of candidate angles, thereby avoiding generic one-dimensional search.
\subsubsection{Quantization-Induced Linearization}
\label{subsec:quant_linearization}
For fixed $m$, over any interval $\mathcal I$ where ${k_m^\star(\phi)}$ stays constant, say ${k_m^\star(\phi)}\equiv k_{m}$, we have
\begin{equation}
\label{eq:quant_piecewise_cos}
\cos\Big(\delta_{M_p}(\phi-\alpha_m)\Big)=\cos\Big(\phi-\alpha_m-\tfrac{2\pi k_m}{M_p}\Big)~(\phi\in\mathcal I),
\end{equation}
since the principal-value difference to the nearest grid point is simply the signed offset inside that cell. Accordingly, on such an $\mathcal I$, each per-port term becomes an ordinary cosine with a fixed phase shift. The $k_m$-switching (boundary) angles for port $m$ occur when $\phi-\alpha_m$ lies on the midpoints between adjacent grid points, which correspond to the midpoints of the vertices of given regular $M_p$-gon:
\begin{equation}
\label{eq:quant_boundary_set}
\Phi_m^{\mathrm{q}}
\triangleq
\Big\{\alpha_m+\tfrac{2\pi k}{M_p}{+}\tfrac{\pi}{M_p}~ (\mathrm{mod}~2\pi): k=0,\cdots,M_p-1\Big\}.
\end{equation}
Let $\Phi^{\mathrm{q}}\triangleq \bigcup_{m=1}^M \Phi_m^{\mathrm{q}}$ denote the global set of quantization breakpoints.
\subsubsection{Top-$M_o$ Switching and Score-Crossing Angles}
\label{subsec:top_switching}
Even if the quantizer indices are fixed, $\Gamma^\star(\phi)$ can switch when the ordering of $\{s_m(\phi)\}$ changes. A switch can only happen at $\phi$ where two scores tie: 
\begin{equation}
\label{eq:score_tie_def}
s_m(\phi)=s_n(\phi)~ (m\neq n).
\end{equation}
Since each $s_m(\phi)$ is piecewise cosine by~\eqref{eq:quant_piecewise_cos}, within any region where all involved ${k_m^\star(\phi)}$ and ${k_n^\star(\phi)}$ are fixed,
\eqref{eq:score_tie_def} reduces to
\begin{equation}
\label{eq:score_tie_cos}
a_m\cos(\phi-\beta_m)=a_n\cos(\phi-\beta_n),
\end{equation}
with $\beta_i\triangleq \alpha_i+\tfrac{2\pi k_i}{M_p}~(i=m,n)$ fixed. Equation~\eqref{eq:score_tie_cos} is equivalent to a linear constraint in $\cos\phi$ and $\sin\phi$:
\begin{equation}
\label{eq:score_tie_linear}
\begin{aligned}
&\big(a_m\cos\beta_m-a_n\cos\beta_n\big)\cos\phi\\
&+\big(a_m\sin\beta_m-a_n\sin\beta_n\big)\sin\phi=0,
\end{aligned}
\end{equation}
{If $a_me^{j\beta_m}\neq a_ne^{j\beta_n}$,~\eqref{eq:score_tie_linear} has exactly two (antipodal) solutions in $[0,2\pi)$, among which only those lying in the region where $(k_m^\star(\phi),k_n^\star(\phi))=(k_m,k_n)$ are retained; if $a_me^{j\beta_m}=a_ne^{j\beta_n}$, the two scores coincide identically on that region, so that their ordering cannot change there. Since there are at most $2M_p$ such regions over $[0,2\pi)$ for each pair $(m,n)$, each pair of scores crosses at most $4M_p$ times.} Collect all such score-crossing angles over all pairs $(m,n)~(m\neq n)$ {and all such regions} into a finite set $\Phi^{\mathrm{top}}$, which correspond to points at which the relative ordering between $s_m(\phi)$ and $s_n(\phi)$ changes and may therefore alter $\Gamma^\star(\phi)$.
\subsubsection{Piecewise Sinusoid Structure and A Finite Maximizer Test}
\label{subsec:piecewise_sinusoid}
Define the global breakpoint set $\Phi_{\mathrm{br}}\triangleq\Phi^{\mathrm{q}}\cup \Phi^{\mathrm{top}}$, and sort them over $[0,2\pi)$ to obtain a partition into finitely many open intervals $\mathcal I_\ell\triangleq(\varphi_\ell,\varphi_{\ell+1})$ and each interval contains no breakpoint.
\begin{lemma}
\label{lem:fixed_structure_region}
On $\mathcal I_\ell$, all $\{{k_m^\star(\phi)}\}_{m=1}^M$ and $\Gamma^\star(\phi)$ remain constant.
\end{lemma}
\begin{proof}
By construction, $\mathcal I_\ell$ contains no quantization breakpoint, so each ${k_m^\star(\phi)}$ cannot change inside $\mathcal I_\ell$, establishing index constancy. Likewise, $\mathcal I_\ell$ contains no score-tie point, so the ordering of $\{s_m(\phi)\}$ cannot change, and hence $\Gamma^\star(\phi)=\mathrm{Top}\text{-}M_o(\{s_m(\phi)\})$ remains constant. 
\end{proof}
On a fixed $\mathcal I_\ell$, denote the constant active set by $\Gamma_\ell$ and the constant quantized phases by $\beta_{m,\ell}\triangleq \alpha_m+\tfrac{2\pi k_{m,\ell}}{M_p}~(m\in\Gamma_\ell)$. Then, by~\eqref{eq:quant_piecewise_cos}, the objective on $\mathcal I_\ell$ becomes a {single sinusoid plus a constant}:
\begin{equation}
\label{eq:region_sinusoid_form}
\begin{aligned}
h_{\mathcal P_{M_o,t}}(e^{j\phi})
&=
a_d\cos(\phi-\alpha_d)
+\sum_{m\in\Gamma_\ell} a_m\cos(\phi-\beta_{m,\ell})\\
%&=\Re\left\{e^{-j\phi}\Big(d+\sum_{m\in\Gamma_\ell} a_m e^{j\beta_{m,\ell}}\Big)\right\}\\
&=\Re\left\{e^{-j\phi} C_\ell\right\}~(\phi\in\mathcal I_\ell),
\end{aligned}
\end{equation}
where we define the region-wise complex resultant
\begin{equation}
\label{eq:Cell_def}
C_\ell \triangleq d+\sum_{m\in\Gamma_\ell} a_m e^{j\beta_{m,\ell}}
=|C_\ell|e^{j\angle C_\ell}.
\end{equation}

\begin{lemma}
\label{lem:regional_maximizer}
{On the closure $\bar{\mathcal I}_\ell=[\varphi_\ell,\varphi_{\ell+1}]$, the maximum of $h_{\mathcal P_{M_o,t}}(e^{j\phi})$ is attained either (i) at $\phi=\angle C_\ell$ if $\angle C_\ell\in\mathcal I_\ell$, or (ii) at an endpoint $\varphi_\ell$ or $\varphi_{\ell+1}$.}
\end{lemma}
\begin{proof}
{Since $h_{\mathcal P_{M_o,t}}(e^{j\phi})$ is the support function of a compact set, it is continuous in $\phi$. Hence, by~\eqref{eq:region_sinusoid_form}, $h_{\mathcal P_{M_o,t}}(e^{j\phi})=|C_\ell|\cos(\phi-\angle C_\ell)$ holds on $\bar{\mathcal I}_\ell$, and its maximum over the compact $\bar{\mathcal I}_\ell$ is attained. If $C_\ell=0$, the function is constant on $\bar{\mathcal I}_\ell$ and the claim is trivial. Otherwise, within a $2\pi$-period, the function has exactly two stationary points, namely the unique maximizer $\angle C_\ell$ and the minimizer $\angle C_\ell+\pi$ (mod $2\pi$). Therefore, a maximizer in the interior of $\bar{\mathcal I}_\ell$ must be $\angle C_\ell$; otherwise, the maximum is attained at an endpoint.} 
\end{proof}
\begin{theorem}
\label{thm:finite_candidate_global}
Define the finite candidate set $\Phi_{\mathrm{cand}}
\triangleq
\Phi_{\mathrm{br}}\cup \Big\{\angle C_\ell : \angle C_\ell\in \mathcal I_\ell \Big\}$. Then {
\begin{equation}
\label{eq:finite_max_rule}
\max_{\phi\in[0,2\pi)} h_{\mathcal P_{M_o,t}}(e^{j\phi})
=
\max_{\phi\in \Phi_{\mathrm{cand}}} h_{\mathcal P_{M_o,t}}(e^{j\phi}),
\end{equation}
and consequently, $\argmax_{\phi\in \Phi_{\mathrm{cand}}} h_{\mathcal P_{M_o,t}}(e^{j\phi})\subseteq\argmax_{\phi\in[0,2\pi)} h_{\mathcal P_{M_o,t}}(e^{j\phi})$, i.e., any maximizer over the finite set $\Phi_{\mathrm{cand}}$ is a global maximizer. Moreover, if $C_\ell\neq0~(\forall\ell)$, every global maximizer belongs to $\Phi_{\mathrm{cand}}$.}
\end{theorem}
\begin{proof}
By Lemma~\ref{lem:fixed_structure_region}, on each $\mathcal I_\ell$ the objective has the sinusoidal form~\eqref{eq:region_sinusoid_form}. By Lemma~\ref{lem:regional_maximizer}, the maximum over ${\bar{\mathcal I}_\ell}$ is achieved either at $\angle C_\ell$ or at an endpoint, and every endpoint belongs to $\Phi_{\mathrm{br}}$. {Since the closures $\{\bar{\mathcal I}_\ell\}$ cover $[0,2\pi)$,} the maximum over each region is attained at some point in $\Phi_{\mathrm{cand}}$. Taking the maximum over all regions yields~\eqref{eq:finite_max_rule}{, and the inclusion follows since $\Phi_{\mathrm{cand}}\subseteq[0,2\pi)$. Finally, a global maximizer not in $\Phi_{\mathrm{br}}$ lies in some $\mathcal I_\ell$ and thus maximizes $h_{\mathcal P_{M_o,t}}(e^{j\phi})$ over $\bar{\mathcal I}_\ell$ in its interior; by the proof of Lemma~\ref{lem:regional_maximizer}, it equals $\angle C_\ell$ when $C_\ell\neq0$}.
\end{proof}
{
\begin{corollary}
\label{cor:max_resultant}
The optimal value of~\eqref{prob:fris_reduced22} satisfies
\begin{equation}
\label{eq:max_resultant}
|z^\star|=\max_{\phi\in[0,2\pi)}h_{\mathcal P_{M_o,t}}(e^{j\phi})=\max_{\ell}|C_\ell|.
\end{equation}
Hence, for $\ell^\star\in\argmax_\ell|C_\ell|$, the configuration $(\Gamma_{\ell^\star},\{k_{m,\ell^\star}\}_{m\in\Gamma_{\ell^\star}})$ is optimal, and $\phi^\star=\angle C_{\ell^\star}$ is a global maximizer of $h_{\mathcal P_{M_o,t}}(e^{j\phi})$.
\end{corollary}
\begin{proof}
Each $C_\ell$ is attained by a feasible configuration, i.e., $C_\ell\in\mathcal Z_{M_o,t}$, and thus $|C_\ell|\le|z^\star|$. Conversely, by~\eqref{eq:equiv_chain_dir} and~\eqref{eq:region_sinusoid_form}, $|z^\star|=\max_{\ell}\max_{\phi\in\bar{\mathcal I}_\ell}|C_\ell|\cos(\phi-\angle C_\ell)\le\max_\ell|C_\ell|$, which proves~\eqref{eq:max_resultant}. Moreover, since $h_{\mathcal P_{M_o,t}}(e^{j\phi})=\max_{z\in\mathcal Z_{M_o,t}}\Re\{e^{-j\phi}z\}\ge\Re\{e^{-j\phi}C_{\ell^\star}\}$ for all $\phi$, setting $\phi=\angle C_{\ell^\star}$ yields $h_{\mathcal P_{M_o,t}}(e^{j\angle C_{\ell^\star}})\ge|C_{\ell^\star}|=|z^\star|$, i.e., the global maximum.
\end{proof}
Since $h_{\mathcal P_{M_o,t}}(e^{j\phi})=\max_{z\in\mathcal Z_{M_o,t}}\Re\{e^{-j\phi}z\}$ is the pointwise maximum of finitely many sinusoids, all of its nondifferentiable points are convex kinks at which the slope increases. Hence, a breakpoint at which $h_{\mathcal P_{M_o,t}}(e^{j\phi})$ is nondifferentiable cannot be a local maximizer, and $\Phi_{\mathrm{br}}$ essentially serves to delimit $\{\mathcal I_\ell\}$.}

Theorem~\ref{thm:finite_candidate_global} {and Corollary~\ref{cor:max_resultant}} reduce the continuous maximization of $h_{\mathcal P_{M_o,t}}(e^{j\phi})$ to evaluating finitely many candidates. {Algorithm~\ref{alg:finite_candidates} implements Corollary~\ref{cor:max_resultant} via an angular sweep that constructs $\{\mathcal I_\ell\}$ together with $\Phi^{\mathrm{top}}$, as detailed in Remark~\ref{rem:phi_top_impl}.}
\begin{remark}
\label{rem:phi_top_impl}
{$\{\mathcal I_\ell\}$ must be formed from $\Phi_{\mathrm{br}}=\Phi^{\mathrm q}\cup\Phi^{\mathrm{top}}$, not from $\Phi^{\mathrm q}$ alone: within an interval delimited only by $\Phi^{\mathrm q}$, the scores are smooth sinusoids that may still cross, so that $\Gamma^\star(\phi)$ may change in its interior. On the other hand, Lemma~\ref{lem:fixed_structure_region}, Theorem~\ref{thm:finite_candidate_global}, and Corollary~\ref{cor:max_resultant} remain valid when $\Phi^{\mathrm{top}}$ contains only the angles at which $\Gamma^\star(\phi)$ actually changes, i.e., ties between the $M_o$th and $(M_o+1)$th largest scores. Such angles can be generated without solving all pairwise crossing equations by an angular sweep that maintains the score ranking as a kinetic sorted list: only ports at adjacent ranks are monitored, and their next crossing angle is obtained in closed-form from~\eqref{eq:score_tie_linear}. A crossing swaps the two ports, while a quantization breakpoint of port $m$ only requires recomputing the crossing angles of $m$ with its two rank neighbors. A swap between the $M_o$th and $(M_o+1)$th ranks changes $\Gamma^\star(\phi)$ and is recorded in $\Phi^{\mathrm{top}}$, and each event updates the resultant $C_\ell$ in $\mathcal O(1)$ operations.}
\end{remark}

\begin{algorithm}[t]
\caption{Maximizing $h_{\mathcal P_{M_o,t}}(e^{j\phi})$ for a Regular $M_p$-gon}
\label{alg:finite_candidates}
\begin{algorithmic}[1]

\Require $\{a_m,\alpha_m\}_{m=1}^M$, $a_d$, $\alpha_d$, $M_o$, $M_p$
\State \multiline{Construct $\Phi^{\mathrm q}=\bigcup_{m=1}^M \Phi_m^{\mathrm q}$ by~\eqref{eq:quant_boundary_set}, and set $\Phi^{\mathrm{top}}\gets\emptyset$.}
%\State \multiline{At $\phi=0$, compute $\{k_m^\star(0)\}$ and $\{s_m(0)\}$ by~\eqref{eq:w_hat_closed_form} and~\eqref{eq:score_closed_form}, respectively, sort $\{s_m(0)\}$, and obtain $\Gamma^\star(0)$ and the resultant $C$ as in~\eqref{eq:Cell_def}.}
\State \multiline{At $\phi=0$, compute $\{s_m(0)\}$ and $\{k_m^\star(0)\}$ by~\eqref{eq:score_closed_form} and~\eqref{eq:w_hat_closed_form}, respectively, sort $\{s_m(0)\}$, and obtain $\Gamma^\star(0)$ and the resultant $C$ as in~\eqref{eq:Cell_def}.}
\State \multiline{Insert the first point of each $\Phi_m^{\mathrm q}$ and the next crossing angles of all adjacent-rank pairs, computed by~\eqref{eq:score_tie_linear}, into an event queue.}
\While{the next event angle $\varphi<2\pi$}
    \State Store $C$ as $C_\ell$ of the region ending at $\varphi$.
    \If{$\varphi\in\Phi_m^{\mathrm q}$ for some $m$}
        \State \multiline{Update $k_m^\star$ (and $C$ if $m\in\Gamma^\star$), insert the next point of $\Phi_m^{\mathrm q}$, and recompute the crossing angles of $m$ with its two rank neighbors.}
    \Else
        \State \multiline{Swap the two tied ports in the ranking and recompute the crossing angles of the new adjacent pairs; if the swap is between the $M_o$th and $(M_o+1)$th ranks, update $\Gamma^\star$ and $C$, and append $\varphi$ to $\Phi^{\mathrm{top}}$.}
    \EndIf
\EndWhile
\State Store $C$ as $C_\ell$ of the last region.
\State \multiline{Output $\ell^\star\in\argmax_\ell|C_\ell|$ and $\phi^\star=\angle C_{\ell^\star}$ by Corollary~\ref{cor:max_resultant}.}
\State \Return $\phi^\star$, $\Gamma_{\ell^\star}$, and $\{k_{m,\ell^\star}\}_{m\in\Gamma_{\ell^\star}}$
\end{algorithmic}
\end{algorithm}
\section{Overall Framework and Complexity Analysis}
\subsection{Overall AO Framework}
\label{sec:AO_framework}
Based on the results, we develop an AO framework to solve~\eqref{prob:fris_exact_dir}. The proposed AO alternates between (A) an update of $\mathbf f_{\mathrm B}$ and (B) an optimal FRIS update. Starting from an initial feasible $\mathbf w^{(0)}$ and with iteration $t$, the AO iterates between the following two subproblems.
\subsubsection{Update of $\mathbf f_{\mathrm B}$}
\label{subsec:AO_stepA}
For a fixed $\mathbf w^{(t)}$, the update of $\mathbf f_{\mathrm B}$:
\begin{equation}
\label{eq:f_update_MRT_AO}
\mathbf f_{\mathrm B}^{(t+1)}
=
\sqrt{P}
\frac{\mathbf a(\mathbf w^{(t)})}{\|\mathbf a(\mathbf w^{(t)})\|_2}.
\end{equation}
\subsubsection{Update of $\mathbf w$ and $\Gamma$}
\label{subsec:AO_stepB}
For a fixed $\mathbf f_{\mathrm B}^{(t+1)}$, we can solve the following FRIS subproblem:
\begin{equation}
\label{prob:w_sub_AO}
\begin{aligned}
&\max_{\{w_m\},\Gamma}~\left|
d^{(t+1)}+\sum_{m\in\Gamma} h_m^{(t+1)} w_m
\right|\\
\text{s.t.}~&
w_m\in\mathcal W_m~(m\in\Gamma),~|\Gamma|=M_o,
\end{aligned}
\end{equation}
where $d^{(t+1)}
\triangleq
\mathbf h_{\mathrm d}^{*}\mathbf f_{\mathrm B}^{(t+1)},~
h_m^{(t+1)}
\triangleq
h_{\mathrm r,m}^{*}\mathbf G_{m,:}\mathbf f_{\mathrm B}^{(t+1)}~(\forall m)$. We apply the proposed framework to obtain an update $(\Gamma^{(t+1)},\mathbf w^{(t+1)})$. The resulting AO procedure is summarized in Algorithm~\ref{alg:AO_joint}, {where note that $\mathbf w$ is first determined, after which $\mathbf f_{\mathrm B}$ is updated.} %Since Section~\ref{subsec:AO_stepA} globally maximizes the objective over $\mathbf f_{\mathrm B}$ for fixed $\mathbf w^{(t)}$, and~\ref{subsec:AO_stepB} optimizes the objective over $(\mathbf w,\Gamma)$ for fixed $\mathbf f_{\mathrm B}^{(t+1)}$, the objective sequence is nondecreasing and upper-bounded by system framework (i.e., $|z|\le |d|+\sum_{m\in\Gamma}|h_m|$), and thus converges.
{The convergence of Algorithm~\ref{alg:AO_joint} is given in following lemma: 
\begin{lemma}
\label{lem:AO_convergence}
Define the objective of~\eqref{prob:fris_exact_dir} as $\Psi(\mathbf f_{\mathrm B},\mathbf w)
\triangleq
\left|
\mathbf h_{\mathrm d}^{*}\mathbf f_{\mathrm B}
+
\mathbf h_{\mathrm r}^{*}\mathrm{diag}(\mathbf w)\mathbf G \mathbf f_{\mathrm B}
\right|$, with feasible set ${\mathcal F}
\triangleq
\left\{
\mathbf w:
w_m\in\Theta_m~(\forall m),~
\big|\{m:w_m\neq 0\}\big|=M_o
\right\}$. Let $\{(\mathbf f_{\mathrm B}^{(t)},\mathbf w^{(t)})\}{_{t\ge1}}$ be the sequence {generated} by Algorithm~\ref{alg:AO_joint}, where
\begin{equation}
\label{ftwt1}
\begin{aligned}
&\mathbf f_{\mathrm B}^{(t+1)}
\in
\argmax_{\|\mathbf f_{\mathrm B}\|_2^2\le P}
\Psi(\mathbf f_{\mathrm B},\mathbf w^{(t)}),\\
&\mathbf w^{(t+1)}
\in
\argmax_{\mathbf w\in{\mathcal F}}
\Psi(\mathbf f_{\mathrm B}^{(t+1)},\mathbf w)~{(t\ge0)}.
\end{aligned}
\end{equation}
Then the objective sequence $\Psi^{(t)}\triangleq \Psi(\mathbf f_{\mathrm B}^{(t)},\mathbf w^{(t)})~{(t\ge1)}$ is monotonically nondecreasing and convergent{, and it becomes constant after finitely many iterations}.
\end{lemma}}
\begin{proof}
{For $t\ge1$,} by Section~\ref{subsec:bf_closed_form}, for fixed $\mathbf w^{(t)}$:
\begin{equation}
\label{eq:pf_step1}
\Psi(\mathbf f_{\mathrm B}^{(t)},\mathbf w^{(t)})
\le
\Psi(\mathbf f_{\mathrm B}^{(t+1)},\mathbf w^{(t)}).
\end{equation}
By Sections~\ref{subsec:rigorous_solution_phi} and~\ref{subsec:regular_Mgon}{, whose FRIS update returns a global maximizer as shown in Corollary~\ref{cor:max_resultant} and Remark~\ref{rem:general_partition}}, for fixed $\mathbf f_{\mathrm B}^{(t+1)}$:
\begin{equation}
\label{eq:pf_step2}
\Psi(\mathbf f_{\mathrm B}^{(t+1)},\mathbf w^{(t)})
\le
\Psi(\mathbf f_{\mathrm B}^{(t+1)},\mathbf w^{(t+1)}).
\end{equation}
Combining~\eqref{eq:pf_step1} and~\eqref{eq:pf_step2} yields $\Psi^{(t)}\le\Psi^{(t+1)}$, so $\{\Psi^{(t)}\}$ is monotonically nondecreasing.

It remains to show boundedness. For any feasible $(\mathbf f_{\mathrm B},\mathbf w)$,
\begin{equation}
\label{eq:pf_bound1}
\Psi(\mathbf f_{\mathrm B},\mathbf w)
=
\left|
\mathbf a(\mathbf w)^* \mathbf f_{\mathrm B}
\right|
\le
\|\mathbf a(\mathbf w)\|_2  \|\mathbf f_{\mathrm B}\|_2
\le
\sqrt{P} \|\mathbf a(\mathbf w)\|_2,
\end{equation}
where $\mathbf a(\mathbf w)
\triangleq
\mathbf h_{\mathrm d}
+
\mathbf G^* \mathrm{diag}(\mathbf w)^* \mathbf h_{\mathrm r}$. Since ${\mathcal F}$ is finite and the channel realization is fixed during optimization, there exists:% a finite constant
\begin{equation}
\label{eq:pf_bound2}
{\bar C}
\triangleq
\sqrt{P}\max_{\mathbf w\in{\mathcal F}}\|\mathbf a(\mathbf w)\|_2
<\infty
\end{equation}
such that $\Psi(\mathbf f_{\mathrm B},\mathbf w)\le {\bar C}$ for all feasible $(\mathbf f_{\mathrm B},\mathbf w)$. Therefore, $\{\Psi^{(t)}\}$ is bounded above{, and hence convergent. Finally, since $\mathbf f_{\mathrm B}^{(t)}$ is determined by $\mathbf w^{(t-1)}$ via~\eqref{eq:f_update_MRT_AO}, $\Psi^{(t)}$ depends only on $(\mathbf w^{(t-1)},\mathbf w^{(t)})\in\mathcal F\times\mathcal F$ and thus takes finitely many values; being nondecreasing, it becomes constant after finitely many iterations}, and the lemma follows.
\end{proof}

\begin{algorithm}[t]
\caption{Proposed AO Framework}
\label{alg:AO_joint}
\begin{algorithmic}[1]
\Require $\mathbf h_{\mathrm d}$, $\mathbf h_{\mathrm r}$, $\mathbf G$, $\{\mathcal W_m\}_{m=1}^M$, $M_o$, $P$, $\epsilon$.
\State Initialize $\mathbf w^{(0)}$ with $|\{m:w_m^{(0)}\neq 0\}|=M_o$; set $t=0$.
\Repeat
    \State Set $\mathbf f_{\mathrm B}^{(t+1)}$ by~\eqref{eq:f_update_MRT_AO}.
    \State Compute $d^{(t+1)}$ and $\{h_m^{(t+1)}\}$.
    \State \multiline{Solve~\eqref{prob:w_sub_AO} to obtain ($\Gamma^{(t+1)}$, $\mathbf w^{(t+1)}$) and the  corresponding $(t+1)$th objective value $|z^{\star(t+1)}|$~(${|z^{\star(0)}|=0}$).}
    \State $t\leftarrow t+1$.
\Until{$||z^{\star(t)}|-|z^{\star(t-1)}|\le\epsilon$}
\State \Return $\mathbf f_{\mathrm B}^{(t)}$, $\Gamma^{(t)}$, $\mathbf w^{(t)}$
\end{algorithmic}
\end{algorithm}
\subsection{Computational Complexity}
\label{ccmp}
%We analyze the computational complexity of the proposed AO framework by examining the per-iteration cost of each update step. 
In each AO iteration, the update of $\mathbf f_{\mathrm B}$ is performed via~\eqref{eq:f_update_MRT_AO}. This step requires computing $\mathbf a$ and its Euclidean norm, which involves matrix-vector multiplications of size $M\times N$ and therefore incurs a complexity of $\mathcal O(MN)$. Subsequently, the computation of $d=\mathbf h_{\mathrm d}^*\mathbf f_{\mathrm B}$ and $\{h_m=h_{\mathrm r,m}^*\mathbf G_{m,:}\mathbf f_{\mathrm B}\}_{m=1}^M$ also requires $\mathcal O(MN)$ operations.

\subsubsection{Complexity Implication for General Codebooks}
The dominant complexity in each AO iteration arises from the FRIS configuration update. Specifically, the proposed Minkowski-geometry-based algorithm evaluates $h_{\mathcal P_{M_o,t}}(e^{j\phi})$ over $N_\phi$ candidate directions{, where $N_\phi$ is the number of intervals of the partition in Algorithm~\ref{alg:geom_support_search_dir} constructed as in Remark~\ref{rem:general_partition}, i.e., $N_\phi=\mathcal O\big(\sum_{m=1}^M|\mathcal W_m|+|\Phi^{\mathrm{top}}|\big)$}. For each $\phi$, $\{s_m(\phi)\}_{m=1}^M$ are computed. In the general finite-codebook case, evaluating each $s_m(\phi)$ requires $\mathcal O(|\mathcal W_m|)$ operations, leading to a total cost of $\mathcal O\left(\sum_{m=1}^M |\mathcal W_m|\right)$ for score computation. The subsequent Top-$M_o$ port selection can be carried out in $\mathcal O(M)$ using selection algorithms, while the optimal codeword assignment for the selected ports can be stored during the score evaluation without additional asymptotic cost. Therefore, the per-direction complexity of the FRIS update is $\mathcal O\left(\sum_{m=1}^M |\mathcal W_m| + M\right)$, and the total FRIS-update complexity per AO iteration scales as $\mathcal O\left(N_\phi\left(\sum_{m=1}^M |\mathcal W_m| + M\right)\right)$. Combining the above steps, the overall computational complexity after $T$ AO iterations is given by $\mathcal O\left(T\left(
MN + N_\phi\left(\sum_{m=1}^M |\mathcal W_m| + M\right)
\right)\right)$. {Alternatively, the partition in Algorithm~\ref{alg:geom_support_search_dir} can be constructed by the angular sweep of Algorithm~\ref{alg:finite_candidates} with the breakpoints in Remark~\ref{rem:general_partition}. Since each pair of scores then crosses at most $2(|\mathcal W_m|+|\mathcal W_n|)$ times and $C$ is updated incrementally, this requires $\mathcal O\big(M\sum_{m=1}^M|\mathcal W_m|\log M\big)$ operations per FRIS update, which leads to the overall complexity of $\mathcal O\big(T\big(MN+M\sum_{m=1}^M|\mathcal W_m|\log M\big)\big)$.}

%In the practically important case of a regular $M_p$-gon phase-shifter codebook, the per-port scores admit closed-form expressions as shown in Section~\ref{subsec:regular_Mgon}. In this case, each $s_m(\phi)$ can be computed in constant time, and the FRIS-update complexity reduces to $\mathcal O(N_\phi M)$ per AO iteration. Accordingly, the overall complexity simplifies to
%\begin{equation}
%\label{ccgon}
%\mathcal O\left(T\left(MN + N_\phi M\right)\right),
%\end{equation}
%which is lower than that of the general finite-codebook case due to the symmetricity of the polygon codebook and the closed-form structure of $s_m(\phi)$.

\subsubsection{Complexity Implication for Regular $M_p$-gon Codebooks}
In the practical case of a regular $M_p$-gon codebook, as a function of $\phi$, $s_m(\phi)$ is piecewise cosine, where its maximizing codeword changes at~\eqref{eq:quant_boundary_set} with {exactly $M_p$} breakpoint angles over $[0,2\pi)$, and {hence $|\Phi^{\mathrm q}|\le MM_p$. Note that $h_{\mathcal P_{M_o,t}}(e^{j\phi})$ is not necessarily smooth between two consecutive points of $\Phi^{\mathrm q}$, since $\Gamma^\star(\phi)$ may change there, as discussed in Remark~\ref{rem:phi_top_impl}. In Algorithm~\ref{alg:finite_candidates}, the initialization requires $\mathcal O(M\log M)$ operations, and each event, i.e., a quantization breakpoint or a crossing of adjacent-rank scores, is processed in $\mathcal O(\log M)$ operations, including the $\mathcal O(1)$ update of $C_\ell$, with a priority queue of size $\mathcal O(M)$, since only the next point of each $\Phi_m^{\mathrm q}$ is kept in the queue. As shown in Section~\ref{subsec:top_switching}, each pair of scores crosses at most $4M_p$ times over $[0,2\pi)$, and thus the number of events is at most $MM_p+4M_p\binom{M}{2}=\mathcal O(M^2M_p)$, which yields a worst-case per-iteration complexity of $\mathcal O(M^2M_p\log M)$. Accordingly, the overall AO complexity is $\mathcal O\left(T\left(MN+M^2M_p\log M\right)\right)$, while in practice only a small fraction of the crossings changes $\Gamma^\star(\phi)$.} As demonstrated in Section~\ref{sec:simulation}, the AO framework converges within only a few iterations, with $T$ typically being much smaller than both $M$ and {$MM_p$}. Therefore, the practical computational burden of the proposed framework is primarily governed by the structural parameters $M$, $N$, {and $M_p$}, rather than by iteration, confirming the scalability and efficiency of the proposed design.

{\subsection{Analysis on Imperfect CSI}
{Acquiring instantaneous CSI in FRIS-assisted systems is challenging since the channel response varies with the selected element positions. Although the design of a dedicated channel-estimation protocol is beyond the scope of this work, several studies provide practical means of obtaining the port-dependent CSI by, e.g., a pilot-training with predefined FRIS configurations~\cite{FRISCE1, FRISCE2}.}%We therefore regard the CSI as an estimate and explicitly analyze the effect of the residual error.}

To quantify the impact of CSI uncertainty, we model the imperfect CSI as bounded perturbations of the effective channel coefficients and establish explicit performance bounds. This starts by the following lemma:
\begin{lemma}
\label{lem:imperfect_csi}
{For a fixed $\mathbf f_{\mathrm B}$,} consider $z(\mathbf w)=d+\sum_{m=1}^{M} w_m h_m$, and suppose that only imperfect CSI is available: $\hat d=d+\Delta d,~\hat h_m=h_m+\Delta h_m~( m=1,\cdots,M)$, where $\Delta d\in\mathbb C$ and $\Delta h_m\in\mathbb C$ denote channel estimation errors. Define $\hat z(\mathbf w)=\hat d+\sum_{m=1}^{M} w_m \hat h_m$. Then the corresponding beamforming-gain mismatch is bounded as
\begin{equation}
\label{eq:gain_difference_bound}
\big||z(\mathbf w)|-|\hat z(\mathbf w)|\big|
\le
|\Delta d|+M_o\max_m |\Delta h_m|.
\end{equation}
\end{lemma}
\begin{proof}
From definition, $z(\mathbf w)-\hat z(\mathbf w)=-\Delta d-\sum_{m=1}^{M} w_m \Delta h_m$. Applying the triangle inequality and using $|w_m|\le 1$ yield
\begin{equation}
\label{pz}
\big|z(\mathbf w)-\hat z(\mathbf w)\big|
\le
|\Delta d|+\sum_{m=1}^{M}|w_m| |\Delta h_m|
\le
|\Delta d|+M_o\max_{m}|\Delta h_m|
\end{equation}
since only $M_o$ coefficients are active. Finally, using the reverse triangle inequality: $\big||z|-|\hat z|\big|\le |z-\hat z|$, we obtain~\eqref{eq:gain_difference_bound}.
\end{proof}
Lemma~\ref{lem:imperfect_csi} shows that the proposed FRIS design is Lipschitz-stable with respect to channel estimation errors. In particular, the degradation of the effective beamforming gain scales at most linearly with the direct-link error and with $M_o$, rather than $M$. This implies that the selection architecture of FRIS is inherently favorable from a robustness perspective: even under imperfect CSI, only the estimation errors of the selected ports directly affect the beamforming-gain mismatch.}

%Furthermore, if $\mathbf w^\star$ denotes the optimizer obtained from imperfect CSI, then its true performance under the actual channel satisfies
%\begin{equation}
%\label{eq:robust_perf_bound}
%|z(\mathbf w^\star)|
%\ge
%|\hat z(\mathbf w^\star)|
%-
%\left(
%|\Delta d|+M_o\max_m|\Delta h_m|
%\right).
%\end{equation}
%Hence, as long as $(|\Delta d|, |\Delta h_m|)$ remains moderate relative to the achieved beamforming gain, the performance degradation is guaranteed to remain limited. This justifies the practical usefulness of the proposed framework under imperfect CSI.

{For the optimality gap analysis under imperfect CSI, we can consider the following corollary:
\begin{corollary}
\label{thm:imperfect_csi_gap}
Suppose that $|\Delta d|\leq\varepsilon_d$ and $\max_m|\Delta h_m|\le\varepsilon_h$~\cite{CE78}, and define $\varepsilon_{\mathrm{CSI}}
\triangleq
\varepsilon_d+M_o\varepsilon_h$. For a true-CSI optimal configuration $\mathbf w^\star
\in
\argmax_{\mathbf w\in{\mathcal F}}|z(\mathbf w)|$, let $\widehat{\mathbf w}^{\star}
\in
\argmax_{\mathbf w\in{\mathcal F}}|\widehat z(\mathbf w)|$ be the FRIS configuration optimized using imperfect CSI. Then the performance gap of the configuration optimized using imperfect CSI $(|z(\widehat{\mathbf w}^{\star})|)$ relative to the true-CSI optimum $(|z(\mathbf w^\star)|)$ is bounded by
\begin{equation}
\label{eq:imperfect_csi_optimality_gap}
0
\leq
|z(\mathbf w^\star)|
-
|z(\widehat{\mathbf w}^{\star})|
\leq
2\varepsilon_{\mathrm{CSI}}.
\end{equation}
\end{corollary}
\begin{proof}
From Lemma~\ref{lem:imperfect_csi}: provides the uniform bound: $||z(\mathbf w)|-|\widehat z(\mathbf w)||
\leq
\varepsilon_{\mathrm{CSI}}~(\forall\mathbf w\in{\mathcal F})$. 
Therefore,
\begin{equation}
\begin{aligned}
\label{fws}
|z(\mathbf w^\star)|-|z(\widehat{\mathbf w}^{\star})|
=&
\bigl[|z(\mathbf w^\star)|-|\widehat z(\mathbf w^\star)|\bigr]
+
\bigl[|\widehat z(\mathbf w^\star)|-|\widehat z(\widehat{\mathbf w}^{\star})|\bigr]
\\
&+
\bigl[|\widehat z(\widehat{\mathbf w}^{\star})|-|z(\widehat{\mathbf w}^{\star})|\bigr].
\end{aligned}
\end{equation}
Since \(\widehat{\mathbf w}^{\star}\) maximizes $ |\widehat z(\mathbf w)|$: $|\widehat z(\mathbf w^\star)|-|\widehat z(\widehat{\mathbf w}^{\star})|\leq0$, and the first and third terms are each upper-bounded by
\(\varepsilon_{\mathrm{CSI}}\). Consequently,
\begin{equation}
\label{ubcor}
|z(\mathbf w^\star)|-|z(\widehat{\mathbf w}^{\star})|
\leq2\varepsilon_{\mathrm{CSI}},
\end{equation}
and the lower-bound follows from the optimality of \(\mathbf w^\star\) for the true-CSI problem.
\end{proof}
Importantly, the result does not require $\widehat{\Gamma}^{\star}
\triangleq
\{m:\widehat w_m^\star\neq0\}=\Gamma^\star
\triangleq
\{m:w_m^\star\neq0\}$. Hence, changes in $\Gamma$ and in the assigned phase codewords are both included in the bound.

We next characterize when the selected-port set remains unchanged under imperfect CSI. For any \(\Gamma\subseteq\{1,\cdots,M\}\) satisfying \(|\Gamma|=M_o\), define the true- and imperfect-CSI support-wise optimal values as
\begin{equation}
\label{eq:support_wise_values}
V_\Gamma
\triangleq
\max_{\substack{\mathbf w\in{\mathcal F}\\
\{m:w_m\neq0\}=\Gamma}}
|z(\mathbf w)|,
~
\widehat V_\Gamma
\triangleq
\max_{\substack{\mathbf w\in{\mathcal F}\\
\{m:w_m\neq0\}=\Gamma}}
|\widehat z(\mathbf w)|,
\end{equation}
respectively. By Lemma~\ref{lem:imperfect_csi}, we can verify that
\begin{equation}
\label{eq:support_value_mismatch}
|V_\Gamma-\widehat V_\Gamma|
\leq
\varepsilon_{\mathrm{CSI}}~(\forall\Gamma: |\Gamma|=M_o).
\end{equation}
Let $\Gamma^\star
\triangleq
\{m:w_m^\star\neq0\}$ denote the unique true-CSI ($z$) optimal port set, and define its support-wise separation margin:
\begin{equation}
\label{eq:support_margin}
\mu_\Gamma
\triangleq
V_{\Gamma^\star}
-
\max_{\substack{\Gamma\neq\Gamma^\star\\|\Gamma|=M_o}}
V_\Gamma.
\end{equation}
If $\mu_\Gamma
>
2\varepsilon_{\mathrm{CSI}}$, then, for every \(\Gamma\neq\Gamma^\star\),
\begin{equation}
\label{vgse}
\widehat V_{\Gamma^\star}
\geq
V_{\Gamma^\star}-\varepsilon_{\mathrm{CSI}}
>
V_\Gamma+\varepsilon_{\mathrm{CSI}}
\geq
\widehat V_\Gamma.
\end{equation}
Therefore, the imperfect-CSI optimization preserves the true \(\Gamma^\star\) when \(\mu_\Gamma>2\varepsilon_{\mathrm{CSI}}\). Otherwise, \(\Gamma^\star\) may change due to estimation errors, while the resulting beamforming-gain loss remains bounded by~\eqref{eq:imperfect_csi_optimality_gap}. Numerical results validating the CSI robustness are presented in Section~\ref{sec:simulation}. {Note that the above bounds pertain to the FRIS subproblem for a given $\mathbf f_{\mathrm B}$; when $\mathbf f_{\mathrm B}$ is also computed from imperfect CSI within the AO, $\varepsilon_d$ and $\varepsilon_h$ are to be interpreted as bounds on the errors of the resulting effective coefficients $d$ and $\{h_m\}$.}
%Therefore, the imperfect-CSI optimization preserves the true $\Gamma^\star$ herein. When $\mu_\Gamma\leq2\varepsilon_{\mathrm{CSI}}$, the estimation errors may alter the ordering of competing port sets, and hence $\Gamma^\star$ may change. Nevertheless, regardless of whether such a change occurs, the true beamforming-gain loss remains bounded as~\eqref{eq:imperfect_csi_optimality_gap}. The numerical verification of CSI robustness will be shown in Section~\ref{sec:simulation}.
}

\section{Simulation Results}
\label{sec:simulation}
%\begin{table}[t]
%\centering
%\caption{Simulation Parameters}
%\label{tab:simulation}
%\begin{tabular}{C{5.8cm} C{1.2cm}}
%\hline
%\textbf{Parameter} & \textbf{Value} \\
%\hline
%Number of BS antennas $N$
%& 16 \\
%\hline
%Number of FRIS elements $M$ (unless referred)
%& $64~(8\times8)$ \\
%\hline
%Number of selected FRIS elements $M_o$\\(unless referred)
%& 8 \\
%\hline
%Transmit SNR $\frac{P}{\sigma^2}$
%& $0$~dB \\
%\hline
%Carrier frequency
%& 3.5~GHz \\
%\hline
%Path-loss exponent
%& 2.5 \\
%\hline
%Rician $K$-factor of $\mathbf G$
%& $3$~dB \\
%\hline
%Resolution of phase codebook $|\mathcal W_m|$ or $M_p$ (unless referred)
%& 8 \\
%\hline
%Normalized FRIS aperture $W_x$ (unless referred)
%& 2 \\
%\hline
%\end{tabular}
%\end{table}
\begin{table}[t]
\centering
\caption{Simulation Parameters (Unless Referred)}
\label{tab:simulation}
\begin{tabular}{C{5.8cm} C{1.2cm}}
\toprule
\textbf{Parameter} & \textbf{Value} \\
\midrule
Number of BS antennas $N$
& 16 \\
{Number of candidate positions} $M$
& $64~(8\times8)$ \\
{Number of fluid reflecting elements} $M_o$
& 8 \\
Transmit SNR $\frac{P}{\sigma^2}$
& $0$~dB \\
Carrier frequency $f_c$ & 3.5~GHz \\
{Bandwidth} & {20~MHz}\\
%{Channel model and related parameters}& {\cite{NR}} \\
Resolution of phase codebook $|\mathcal W_m|$ or $M_p$
& 8 \\
Normalized FRIS aperture $W_x$ (unless referred)
& 2 \\
\bottomrule
\end{tabular}
\end{table}

%\begin{table}[t]
%\centering
%\caption{{Simulation Parameters of the Channel Model}}
%\label{tab:3gpp_channel}
%%\footnotesize
%%\setlength{\tabcolsep}{4pt}
%%\renewcommand{\arraystretch}{1.15}
%\begin{tabular}{C{1.95cm}C{2.5cm}C{3.2cm}}
%\toprule
%\textbf{Parameter}
%& \textbf{LOS}
%& \textbf{NLOS} \\
%\midrule
%
%Path-loss [dB] (${d_{\mathrm{3D}}}$:~[m], $f_c$:~[GHz])
%&
%$\mathrm{PL}_{\mathrm{LOS}}
%=32.4+21\log_{10}{d_{\mathrm{3D}}}
% +20\log_{10}f_{\mathrm c}$
%&
%$\begin{aligned}
%\mathrm{PL}_{\mathrm{NLOS}}&
%=\max\{\mathrm{PL}_{\mathrm{LOS}},\\
%&35.3\log_{10}{d_{\mathrm{3D}}}+22.4\\
%&+21.3\log_{10}f_{\mathrm c}\}
%\end{aligned}$ \\
%
%Number of clusters
%& 12
%& 19 \\
%
%Rays per cluster
%& 20
%& 20 \\
%
%Shadow-fading standard dev.
%& 4~dB
%& 7.82~dB \\
%
%Per-cluster shadowing
%& 3~dB
%& 3~dB \\
%
%Delay spread ($\log_{10} ({\rm{DS}}/{\rm 1s})$)
%& $\mathcal N(-7.297, 0.380^2)$
%& $\mathcal N(-6.987, 0.385^2)$ \\
%
%Delay scaling parameter
%& 3
%& 2.1 \\
%
%\bottomrule
%\end{tabular}
%\end{table}
We evaluate the performance of the proposed FRIS configuration algorithm, where the system parameters are in Table~\ref{tab:simulation}. Unless otherwise specified, we consider an FRIS-aided downlink system assisted by BS, FRIS, user located in $(0, 0, 5)$~m, $(10, 10, 5)$~m, $(50, 0, 0)$~m, respectively. We adopt the 3GPP TR~38.901 UMi-Street-Canyon model~\cite{NR}, 
%including number of clusters and paths per clusters, path-loss models, and related delay and shadowing parameters.\footnote{The 3GPP model is used to generate $\widetilde{\mathbf G}$, $\widetilde{\mathbf h}_{\mathrm r}$, and $\mathbf h_{\mathrm d}$, while the spatial correlation among the FRIS ports is incorporated solely through $\mathbf R$.} %with the selected main parameters summarized in Table~\ref{tab:3gpp_channel}.
%In particular, the path-loss and shadow fading follow~\cite[Table~7.4.1-1]{NR}, while the cluster delays, powers, angular spreads, and subpath parameters follow~\cite[Tables~7.5-3, 7.5-6, and 7.5-8]{NR}. 
and the resulting wideband multipath channel is evaluated at $f_c$ to obtain the frequency-flat channel coefficients used in the proposed framework. Each selected FRIS port applies a discrete phase-shifting coefficient chosen from a finite unit-modulus codebook, where every $\arg w_m~(w_m\in\mathcal W_m)$ is uniformly selected from $[0, 2\pi)$.

For performance verification, the achievable rate is evaluated as $R = \log_2\left(1 + \frac{ |z|^2}{\sigma^2}\right)$, which is a strictly increasing function of $|z|$ and hence any improvement in $|z|$ directly translates into a monotonic increase in $R$. 
%For the multi-user case, after obtaining a given FRIS configuration \(\mathbf w^\star\), the ZF precoder is recomputed as
%\begin{equation}
%\label{reczf}
%\mathbf F_{\mathrm{ZF}}(\mathbf w^\star)
%=
%\sqrt{\eta(\mathbf w)}
%\mathbf H_{\mathrm{eff}}^*(\mathbf w^\star)
%\left(
%\mathbf H_{\mathrm{eff}}(\mathbf w^\star)
%\mathbf H_{\mathrm{eff}}^*(\mathbf w^\star)
%\right)^{-1}.
%\end{equation}
%Consequently, $\mathbf H_{\mathrm{eff}}(\mathbf w)
%\mathbf F_{\mathrm{ZF}}(\mathbf w)
%=
%\sqrt{\eta(\mathbf w)}\mathbf I_K$, so that \(z_{k,i}=0\) for \(i\neq k\) and
%\(|z_{k,k}|=\sqrt{\eta(\mathbf w)}\) for every user. The achievable multi-user sum rate is therefore evaluated as
%\begin{equation}
%\label{mus}
%R
%=
%\sum_{k=1}^{K}
%\log_2\left(
%1+\frac{|z_{k,k}|^2}{\sigma^2}
%\right)
%=
%K\log_2\left(
%1+\frac{\eta(\mathbf w)}{\sigma^2}
%\right).
%\end{equation}
%Hence, any monotonic improvement in the multi-user ZF gain ($\{|z_{k,k}|^2\}$) directly translates into a monotonic increase in $R$.} 
 All reported results are averaged over $10^3$ independent channel realizations, and we evaluate the average achievable rate $\mathbb{E}[R]$. For optimality-verification plots, we additionally consider the normalized beamforming-gain ratio: $\Delta_G
\triangleq
\frac{|z|}{|z_{\rm opt}|}$, where $z_{\rm opt}$ corresponds to the globally optimal solution obtained via exhaustive search of $M_o$ FRIS ports among $M$ candidates. Due to the exponential complexity of brute-force search, $\Delta_G$ is evaluated for $M=9$ and $M_o\le 4$.

{We compare the proposed algorithm with following benchmark schemes:
\begin{enumerate}
\item \textbf{``CEM-based''}: scheme on~\cite{FRISonoff}, based on cross-entropy method (CEM).
\item \textbf{``Random ports''}: $\Gamma^\star$ is selected uniformly from all $\binom{M}{M_o}$ possibilities. For the selected ports, $\{w_m^\star\}_{m\in\Gamma^\star}$ are chosen by {$w_m^\star\in\argmax_{w\in\mathcal W_m}\Re\{e^{-j\arg d}h_mw\}~(\forall m\in\Gamma^\star)$, i.e., the codeword whose phase is closest to $\arg d-\arg h_m$}~\cite{JAP, DF}.
\item \textbf{``Top-$|h|$ Ports''}: the $M_o$ ports with the largest $|h_m|$ are deterministically selected to form $\Gamma^\star$, and the corresponding $\{w_m^\star\}_{m\in\Gamma^\star}$ are determined in the same quantized phase-alignment manner as in the random-port scheme.
\item \textbf{``RIS''}: conventional RIS optimization on~\cite{JAP, DF} {with $M_o$ reflection elements}.
\end{enumerate}
Note that for the second and third schemes, $\mathbf f_{\mathrm B}$ is determined by same means of the proposed framework. }

\begin{figure}[t]
  \begin{center}
    \includegraphics[width=0.65\columnwidth,keepaspectratio]{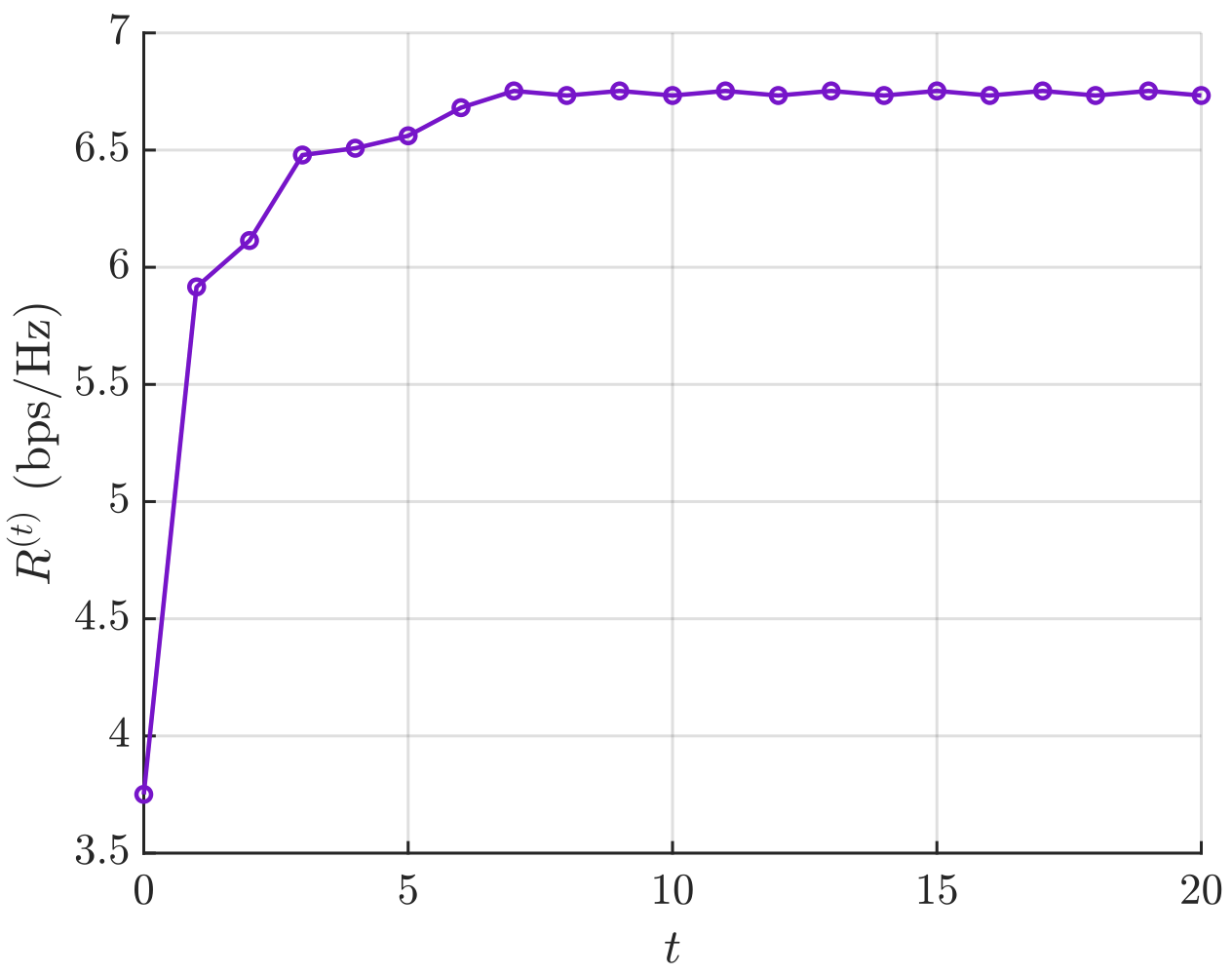}
    \caption{Rate $R^{(t)}$ according to iteration $t$ under general codebook.}
    \label{fig_iter}
  \end{center}
\end{figure}
%\begin{figure}[t]
%    \centering
%    \subfloat[]{%
%        \includegraphics[width=0.24\textwidth]{Fig/iter}
%        \label{fig_iter}%
%    }
%    \subfloat[]{%
%        \includegraphics[width=0.24\textwidth]{Fig/iterk}
%        \label{fig_iterk}%
%    }
%    \caption{{Rate $R^{(t)}$ according to iteration $t$ under general codebook with (a) single- and (b) multi-user scenario.}}
%    \label{fig_itert}
%\end{figure}

{Fig.~\ref{fig_iter} illustrates the convergence behavior of the proposed framework by showing the evolution of the achievable rate $R^{(t)}$ with respect to the iteration index $t$. Herein, $R^{(t)}$ increases rapidly and saturates within only a few iterations, indicating fast convergence; hence, $T$ in Section~\ref{ccmp} remains very small, typically fewer than ten iterations.%, consistent with the theoretical properties of the proposed AO procedure, and numerically guarantees the convergence of the multi-user scenario together with~\eqref{eq:multi_monotonic_chain} and the optimization of $\mathbf F^{(t)}$. 
This trend is numerically consistent with the theoretical properties of the proposed AO procedure, in which each iteration monotonically improves (or preserves) the objective value.} %The minor fluctuations observed around the converged value stem from the discrete FRIS port selection and phase quantization, but their impact is negligible, thereby confirming the stability and efficiency of the proposed framework.}

%{\begin{table}[t]
%\centering
%\caption{{Normalized Average Achievable Rate with\\the First-Order Result Set to 100\%}}
%\label{tab:zero}
%\begin{tabular}{lcc}
%\toprule
%\textbf{Codebook} 
%& \textbf{Zeroth-order} 
%& \textbf{First-order} \\
%\midrule
%General codebook
%& $99.680\%$
%& $100\%$ \\
%Regular $M_p$-gon codebook
%& $99.963\%$
%& $100\%$ \\
%\bottomrule
%\end{tabular}
%\end{table}
%Table~\ref{tab:zero} compares the normalized average achievable rates obtained using the zeroth-($R^{(0)}$) and first-($R^{(1)}$) order Neumann approximations of $\boldsymbol\Phi$ with $f_c=27.5$~GHz. For each codebook, the first-order performance is normalized to $100\%$, while the zeroth-order result is calculated as $\frac{\mathbb{E}[R^{(0)}]}{\mathbb{E}[R^{(1)}]}$. Herein, $\mathbf S$ used in this evaluation is constructed using the scattering-parameter values reported in~\cite[Table~I]{r1r1} for an RIS prototype. As shown in Table~\ref{tab:zero}, the zeroth-order model retains 99.680\% and
%99.963\% of the corresponding first-order performance under the general and regular $M_p$-gon codebooks, respectively. These results indicate that the proposed zeroth-order Minkowski-based design experiences only a marginal performance discrepancy under the considered mutual-coupling realization. Nevertheless, since the scattering parameters depend on the operating frequency and hardware geometry, the results should be interpreted as a reference mutual-coupling sensitivity evaluation.}

\begin{figure}[t]
    \centering
    \subfloat[]{%
        \includegraphics[width=0.24\textwidth]{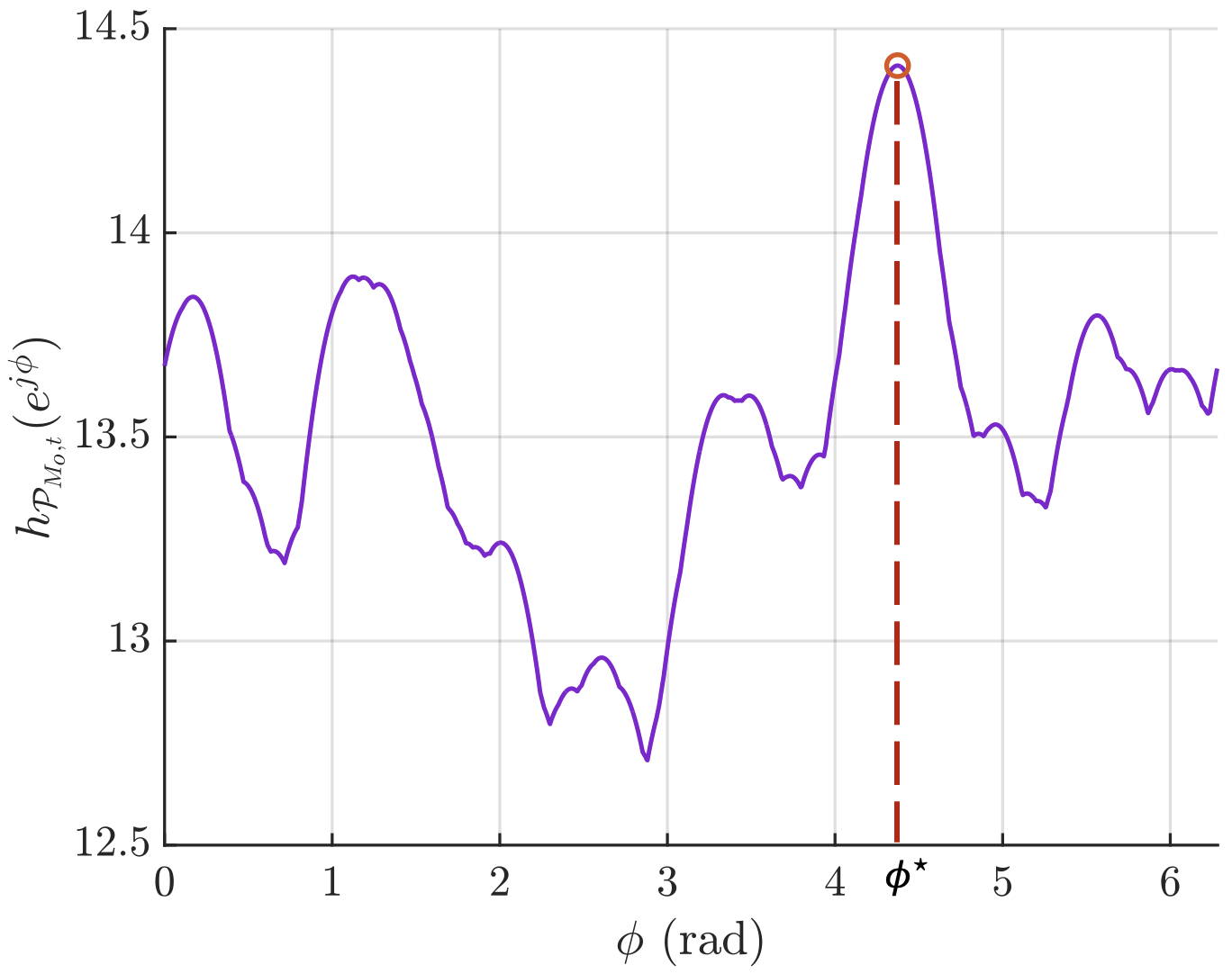}
        \label{fig_hphi1}%
    }
    \subfloat[]{%
        \includegraphics[width=0.24\textwidth]{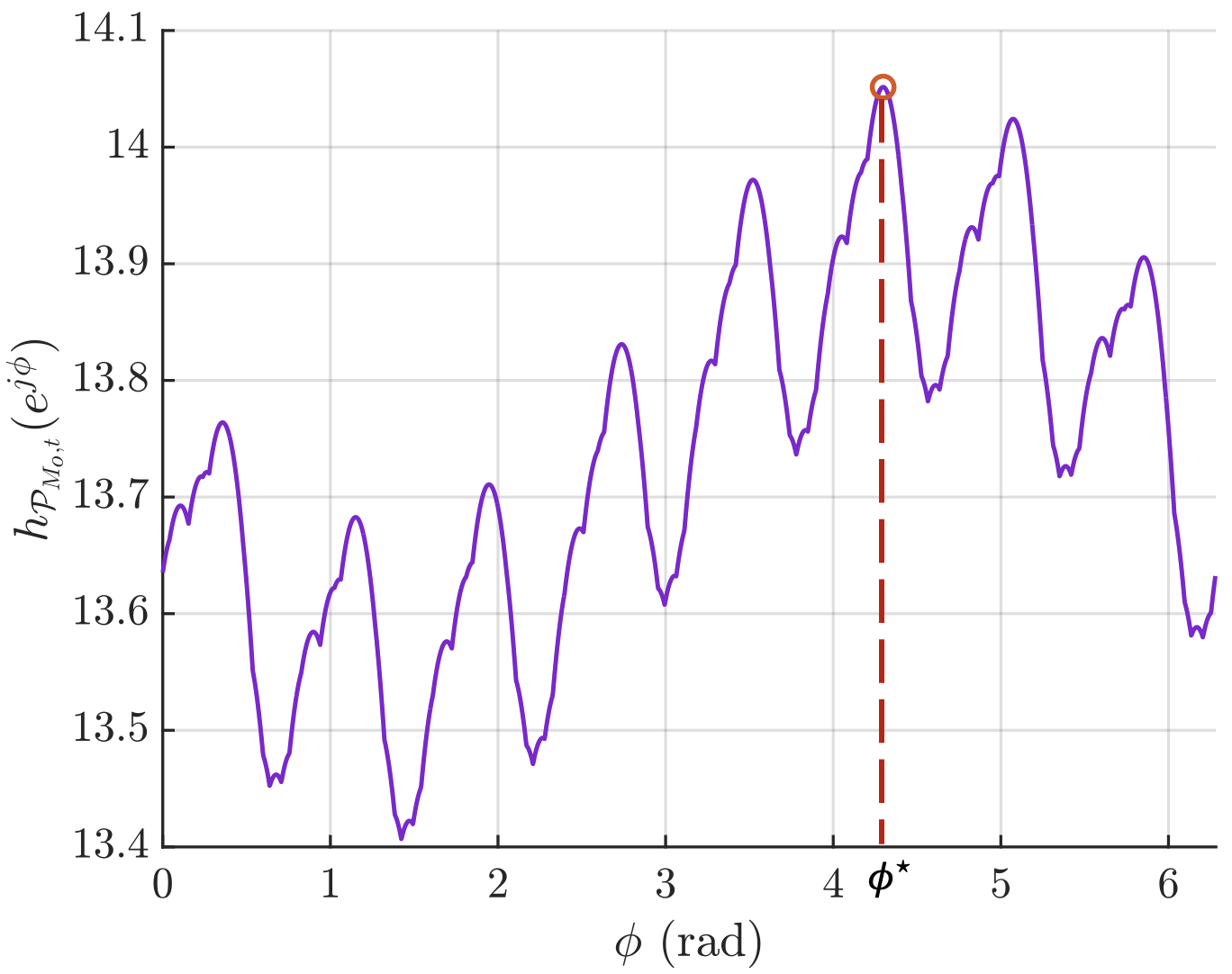}
        \label{fig_hphi2}%
    }
    \caption{$h_{\mathcal P_{M_o,t}}(e^{j\phi})$ according to $\phi\in[0,2\pi)$ with (a) general and (b) regular $M_p$-gon codebook.}
    \label{fig_hphi}
\end{figure}

Fig.~\ref{fig_hphi} illustrates $h_{\mathcal P_{M_o,t}}(e^{j\phi})$ as a function of $\phi$. In the case of a general finite codebook in Fig.~\ref{fig_hphi1}, $h_{\mathcal P_{M_o,t}}(e^{j\phi})$ exhibits an irregular but continuous profile, reflecting the combined effect of port selection and discrete phase adaptation across directions. By contrast, under a regular $M_p$-gon codebook in Fig.~\ref{fig_hphi2}, it becomes more regular and structured due to uniform phase quantization and its closed-form cosine structure in~\eqref{eq:phi_obj_closed_form_dir}. Herein, within each interval where $\Gamma^\star(\phi)$ and $k_m^\star(\phi)$ remain unchanged, the function reduces to a smooth sum of cosine terms and admits at most one {local maximizer}. Non-smooth points occur only at a finite number of breakpoints induced by $\Phi_{\mathrm{br}}$, which correspond to the visible kinks in Fig.~\ref{fig_hphi}{, all of which are convex kinks as shown in Section~\ref{subsec:piecewise_sinusoid}}. In both cases, $\phi^\star$ corresponds to the global maximum of $h_{\mathcal P_{M_o,t}}(e^{j\phi})$, which directly determines the optimal FRIS configuration. This observation validates the proposed one-dimensional directional search and highlights how finite phase codebooks shape the geometry of the support function.

\begin{figure}[t]
    \centering
    \subfloat[]{%
        \includegraphics[width=0.24\textwidth]{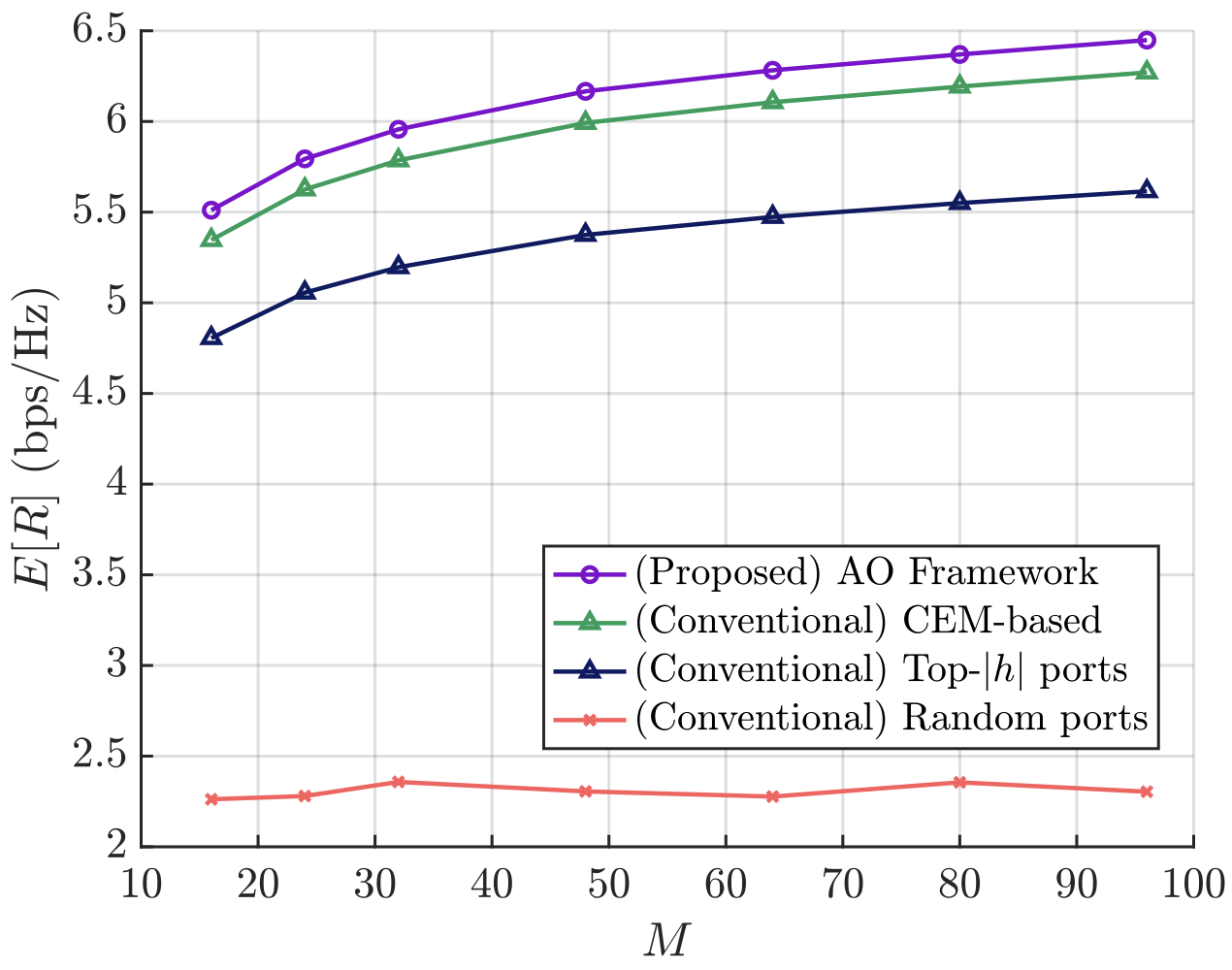}
        \label{fig_rm}%
    }
    \subfloat[]{%
        \includegraphics[width=0.24\textwidth]{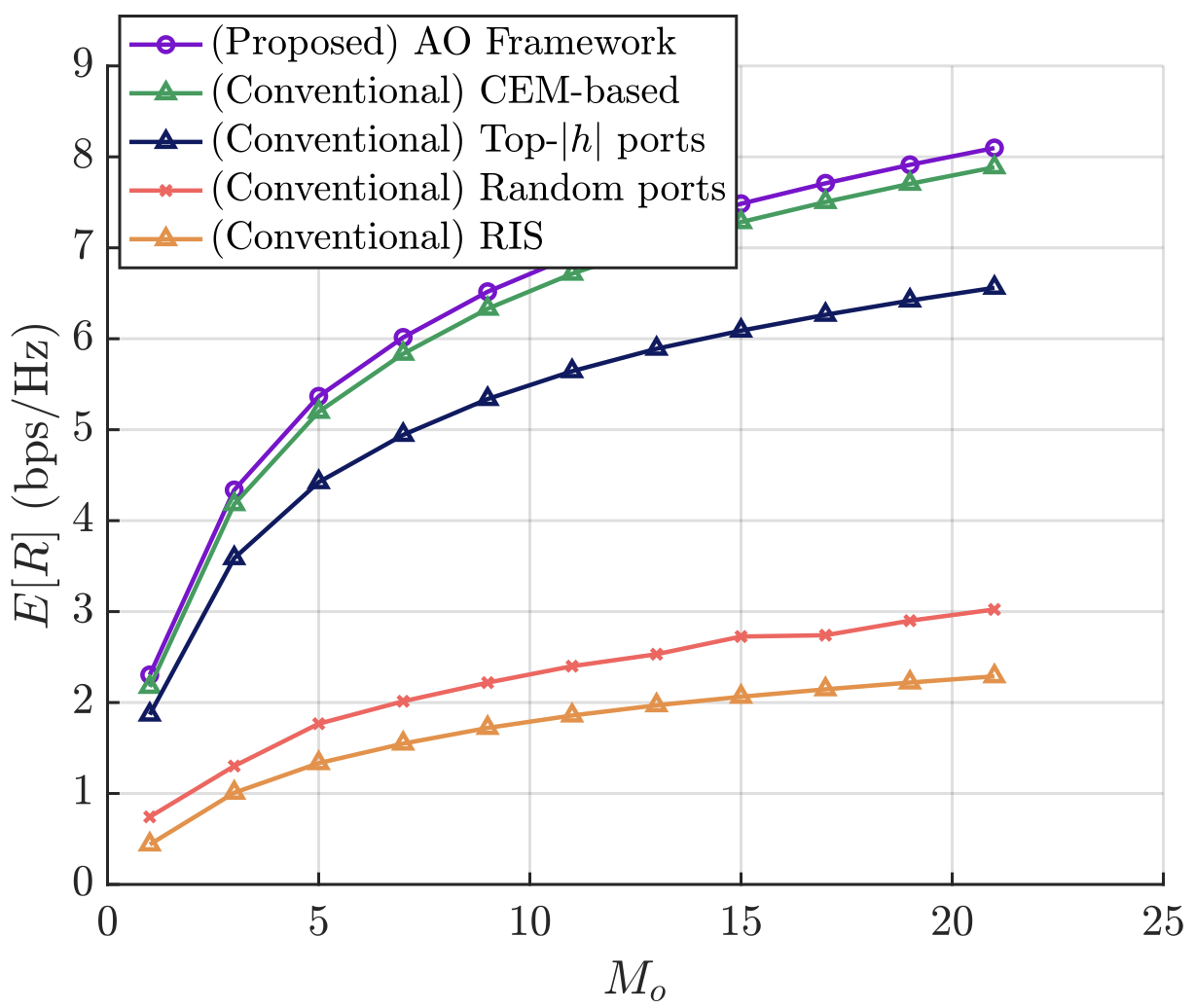}
        \label{fig_rmo}%
    }
%    \vfill
%        \subfloat[]{%
%        \includegraphics[width=0.24\textwidth]{Fig/rk}
%        \label{fig_rk}%
%    }
    \caption{$\mathbb E[R]$ according to (a) $M$ and (b) $M_o$ under general codebook.}
    \label{fig_1}
\end{figure}
%\begin{figure}[t]
%  \begin{center}
%    \includegraphics[width=0.5\columnwidth,keepaspectratio]{Fig/rm}
%    \caption{$\mathbb E[R]$ according to $M$ under general codebook.}
%    \label{fig_rm}
%  \end{center}
%\end{figure}
%\begin{figure}[t]
%  \begin{center}
%    \includegraphics[width=0.5\columnwidth,keepaspectratio]{Fig/rmo}
%    \caption{$\mathbb E[R]$ according to $M_o$ under general codebook.}
%    \label{fig_rmo}
%  \end{center}
%\end{figure}

{Fig.~\ref{fig_rm} illustrates $\mathbb{E}[R]$ as a function of $M$. As $M$ increases, $\mathbb{E}[R]$ improves monotonically for all schemes, since a larger $M$ provides more candidate location in selecting favorable FRIS ports. The proposed FRIS design benefits most from this increase, as the enlarged port pool allows for more effective directional scoring and the selection of the Top-$M_o$ ports that maximize $s_m(\phi)$. This facilitates the identification of ports that are optimally aligned for coherent combining, thereby yielding higher beamforming gains. In contrast, the random-port and Top-$|h|$ benchmarks rely on non-geometric or magnitude-only selection rules, which cannot fully utilize the increased spatial diversity. Though the CEM-based scheme adopts stochastic sampling and iterative updates, its performance is observed to be slightly inferior since the sampling-based nature may not fully exploit the underlying geometric structure and can introduce variability in the solution quality, particularly as the search space expands with larger $M$~\cite{CEM}. In contrast, the proposed deterministic framework directly captures the dominant combining direction, leading to a more consistent and effective utilization of the enlarged candidate set. As a result, the performance gap between the proposed framework and these benchmarks becomes more pronounced as $M$ grows.} %This is because increasing $M$ enlarges the search space of candidate ports and thus creates more multiport combining opportunities, and the proposed framework can more effectively exploit the additional spatial DoF through joint port selection and quantized phase adaptation. In contrast, the benchmarks use either random selection/sampling (i.e., CEM) or a myopic magnitude-based rule (i.e., Top-$|h|$), which cannot fully capture inter-port complementarity among search space that becomes increasingly important as $M$ increases.}

Fig.~\ref{fig_rmo} shows $\mathbb{E}[R]$ versus $M_o$. Increasing $M_o$ leads to a consistent performance improvement, since more physical reflecting elements contribute constructively to $z=d+\sum_{m\in\Gamma}h_m w_m$. %\footnote{{Thus, the result should not be interpreted as suggesting that partial deployment is preferable to a fully populated RIS. Rather, it quantifies the performance obtained under different reflecting-element budgets $(M_o)$. For every prescribed \(M_o\), the proposed FRIS exploits position reconfigurability to outperform a fixed-position RIS with same $M_o$ controllable elements.}} 
The proposed Minkowski-geometry-based framework achieves the largest gain by selecting the $M_o$ ports with the highest $s_m(\phi)$, ensuring near-optimal coherent combining. The CEM-based approach achieves competitive performance but may exhibit slightly lower gains due to its sampling-based nature, which does not fully exploit the underlying geometric structure and may introduce variability in the solution quality, particularly as $M_o$ increases. By contrast, the Top-$|h|$ scheme does not explicitly account for phase alignment across ports, while the random-port scheme selects ports without considering their geometric contribution, resulting in noticeably lower performance, especially for larger $M_o$. The conventional RIS performs the worst, as it relies on a fixed locations without position reconfigurability and therefore cannot leverage the increased $M_o$ for coherent combining.

%{Fig.~\ref{fig_rk} depicts $\mathbb E[R]$ versus $K$ with applying the framework in Section~\ref{sec:multiuser_extension}. As $K$ increases, $\mathbb E[R]$ initially improves due to the multiplexing gain, since more users are simultaneously served by the BS. However, beyond a certain point, the performance saturates and eventually decreases. This behavior arises from the fundamental trade-off in ZF-based multi-user systems: while increasing $K$ enhances spatial multiplexing, it simultaneously reduces the effective array gain per user and increases the sensitivity to noise amplification, as reflected in the $N-K$ term in the SINR expression~\cite{mcj}. Consequently, $\mathbb E[R]$ exhibits a concave trend with respect to $K$, though the proposed framework consistently achieves the highest performance across all considered regimes.}

\begin{figure}[t]
    \centering
    \subfloat[]{%
        \includegraphics[width=0.24\textwidth]{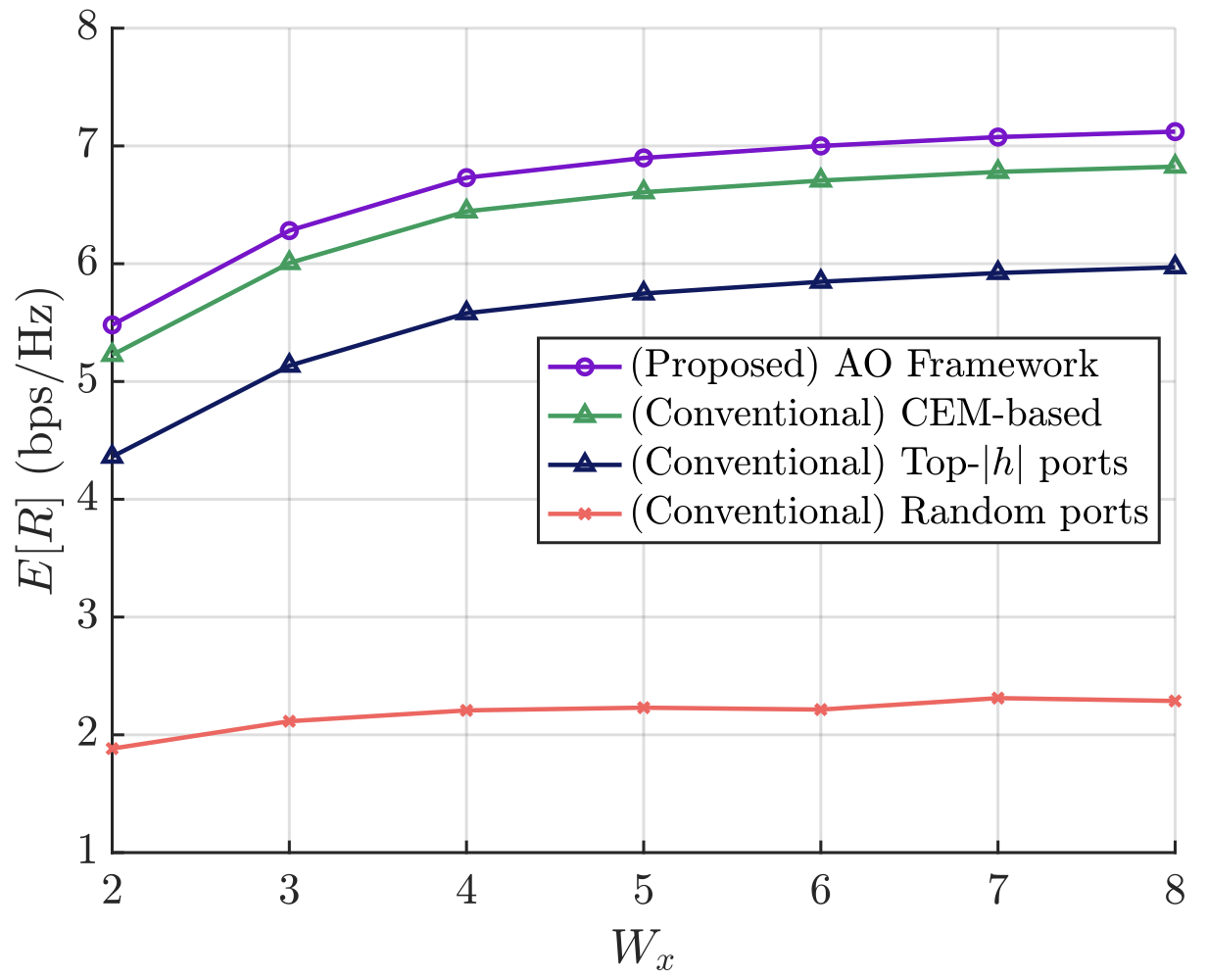}
        \label{fig_rwx}%
    }
    \subfloat[]{%
        \includegraphics[width=0.24\textwidth]{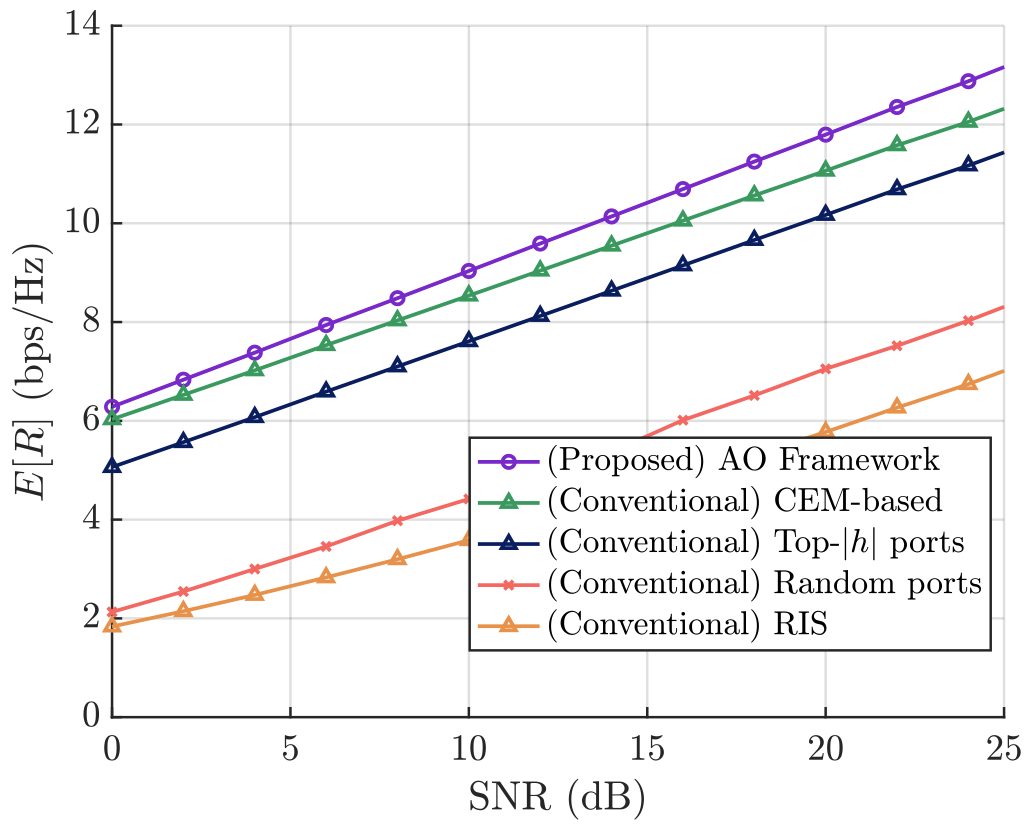}
        \label{fig_rsnr}%
    }
    \vfill
        \subfloat[]{%
        \includegraphics[width=0.24\textwidth]{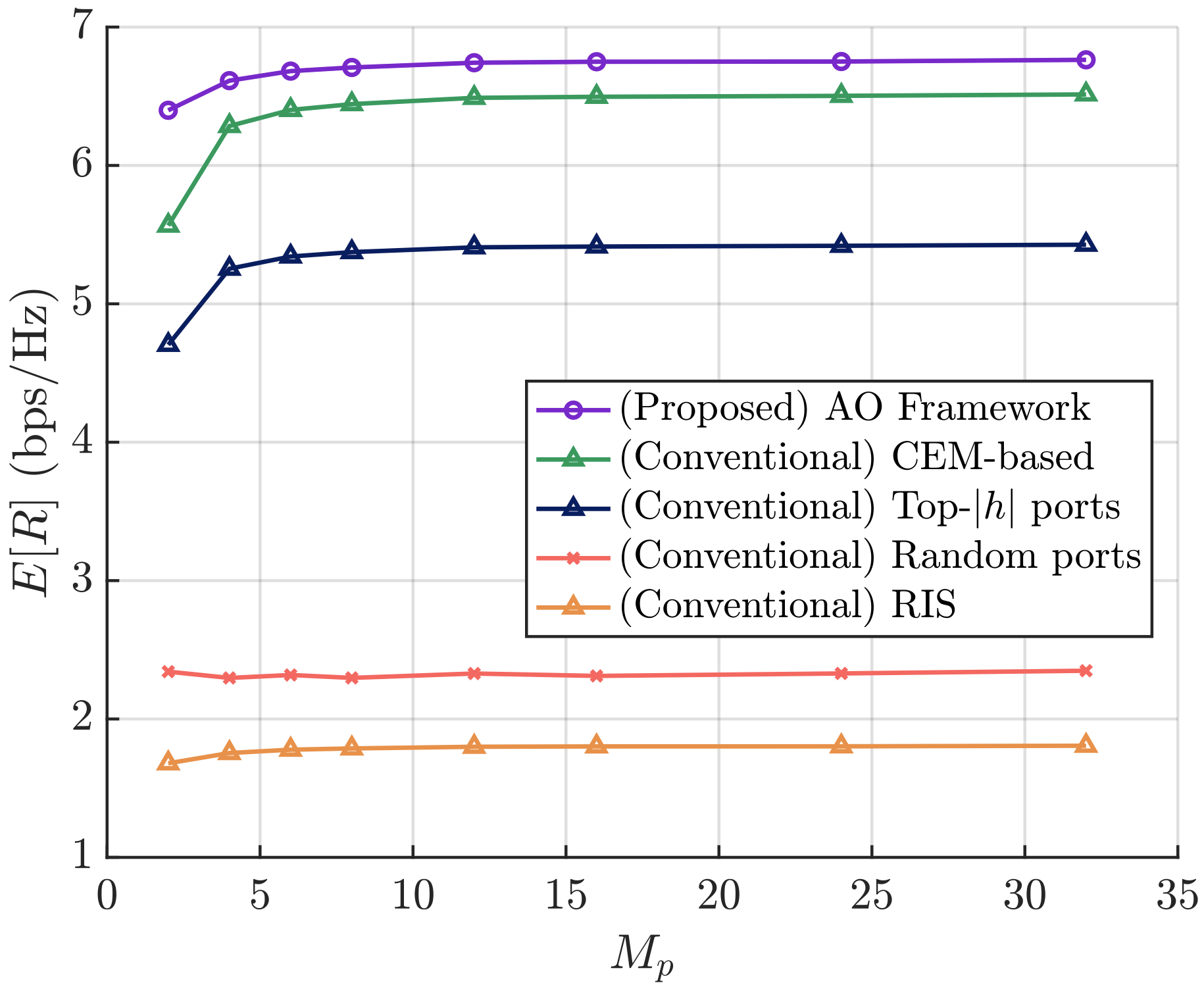}
        \label{fig_rmp}%
    }
    \caption{$\mathbb E[R]$ according to (a) $W_x$ and (b) $\mathrm{SNR}$ under general codebook and (c) $M_{\mathrm p}$ under $M_{\mathrm p}$-gon codebook.}
    \label{fig_2}
\end{figure}
%\begin{figure}[t]
%  \begin{center}
%    \includegraphics[width=0.5\columnwidth,keepaspectratio]{Fig/rmp}
%    \caption{$\mathbb E[R]$ according to $M_{\mathrm p}$ under $M_{\mathrm p}$-gon codebook.}
%    \label{fig_rmp}
%  \end{center}
%\end{figure}
%\begin{figure}[t]
%  \begin{center}
%    \includegraphics[width=0.5\columnwidth,keepaspectratio]{Fig/rwx}
%    \caption{$\mathbb E[R]$ according to $W_x$ under general codebook.}
%    \label{fig_rwx}
%  \end{center}
%\end{figure}
{Fig.~\ref{fig_rwx} illustrates $\mathbb{E}[R]$ versus $W_x$. As $W_x$ increases, the inter-element spacing grows and the spatial correlation among candidate ports decreases, providing richer location diversity and thus improving the beamforming gain. Accordingly, $\mathbb{E}[R]$ increases and gradually saturates for sufficiently large $W_x$, since after ${d_{\mathrm s}}=\frac{\lambda}{2}$-spacing, the spatial correlation characterized by $j_0$, {whose first null occurs at $\frac{\lambda}{2}$}~\cite{perlimfas}, i.e., $W_x>4~(\because M_x=8)$, is already sufficiently mitigated. Herein, the proposed framework consistently achieves the highest $\mathbb{E}[R]$, whereas the benchmarks cannot fully exploit the additional spatial DoF.}

%\begin{figure}[t]
%  \begin{center}
%    \includegraphics[width=0.5\columnwidth,keepaspectratio]{Fig/rsnr}
%    \caption{$\mathbb E[R]$ according to $\mathrm{SNR}$ under general codebook.}
%    \label{fig_rsnr}
%  \end{center}
%\end{figure}

{Fig.~\ref{fig_rsnr} presents $\mathbb{E}[R]$ as a function of signal-to-noise-ratio $(\mathrm{SNR})$. Herein, the proposed framework consistently outperforms the benchmark schemes across the entire $\mathrm{SNR}$ range. This advantage stems from the joint optimization of beamforming and FRIS configuration, which maximizes $|z|$ and thus preserves a strong array gain even at high $\mathrm{SNR}$. In contrast, the benchmark schemes suffer from suboptimal port selection, phase misalignment, or the absence on position reconfigurability, leading to a persistent performance gap.}

%\begin{figure}[t]
%  \begin{center}
%    \includegraphics[width=0.5\columnwidth,keepaspectratio]{Fig/rmp}
%    \caption{$\mathbb E[R]$ according to $M_{\mathrm p}$ under $M_{\mathrm p}$-gon codebook.}
%    \label{fig_rmp}
%  \end{center}
%\end{figure}

{Fig.~\ref{fig_rmp} depicts the impact of $M_p$ on $\mathbb{E}[R]$ under a regular $M_p$-gon codebook. As $M_p$ increases, $\mathbb{E}[R]$ improves and gradually saturates. This behavior is explained by the reduction of phase-quantization error as $M_p$ grows: finer phase resolution allows each FRIS port to better align its reflected signal with the optimal direction. According to~\eqref{eq:delta_def} and~\eqref{eq:score_closed_form}, $\delta_{M_p}$ decreases with $M_p$, leading to larger effective combining gains. Once the quantization error becomes sufficiently small, further increases in $M_p$ yield diminishing returns, resulting in the observed saturation. Nevertheless, the proposed framework consistently outperforms the benchmarks, since it jointly optimizes port selection and quantized phase adaptation in a direction-aware manner based on Minkowski-sum-geometry of codebook spaces, whereas the benchmarks rely on stochastic or port-wise rules or fixed-position structure that cannot fully exploit the reduced error even at high phase resolutions.}

\begin{figure}[t]
    \centering
    \subfloat[]{%
        \includegraphics[width=0.24\textwidth]{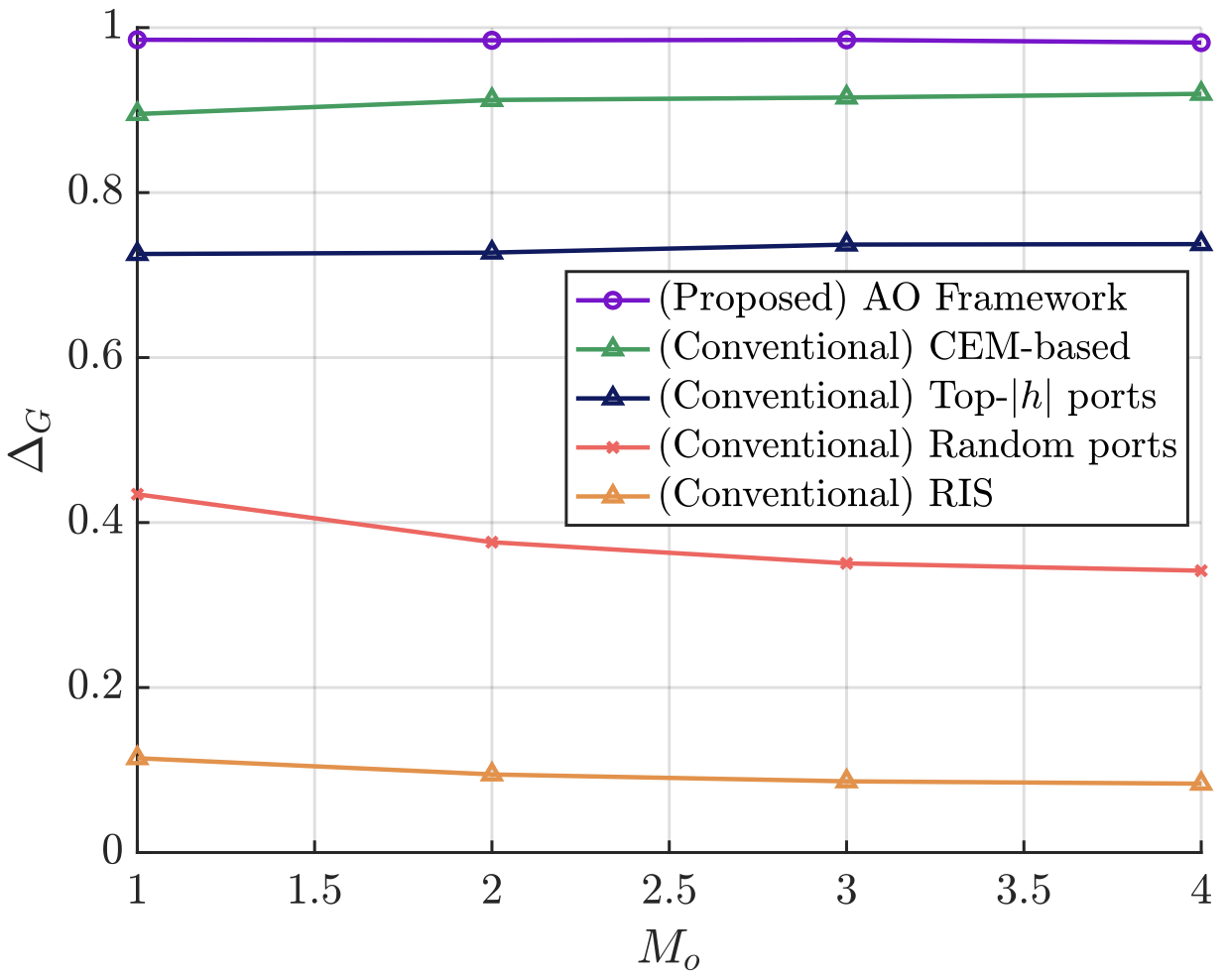}
        \label{fig_dmo}%
    }
    \subfloat[]{%
        \includegraphics[width=0.24\textwidth]{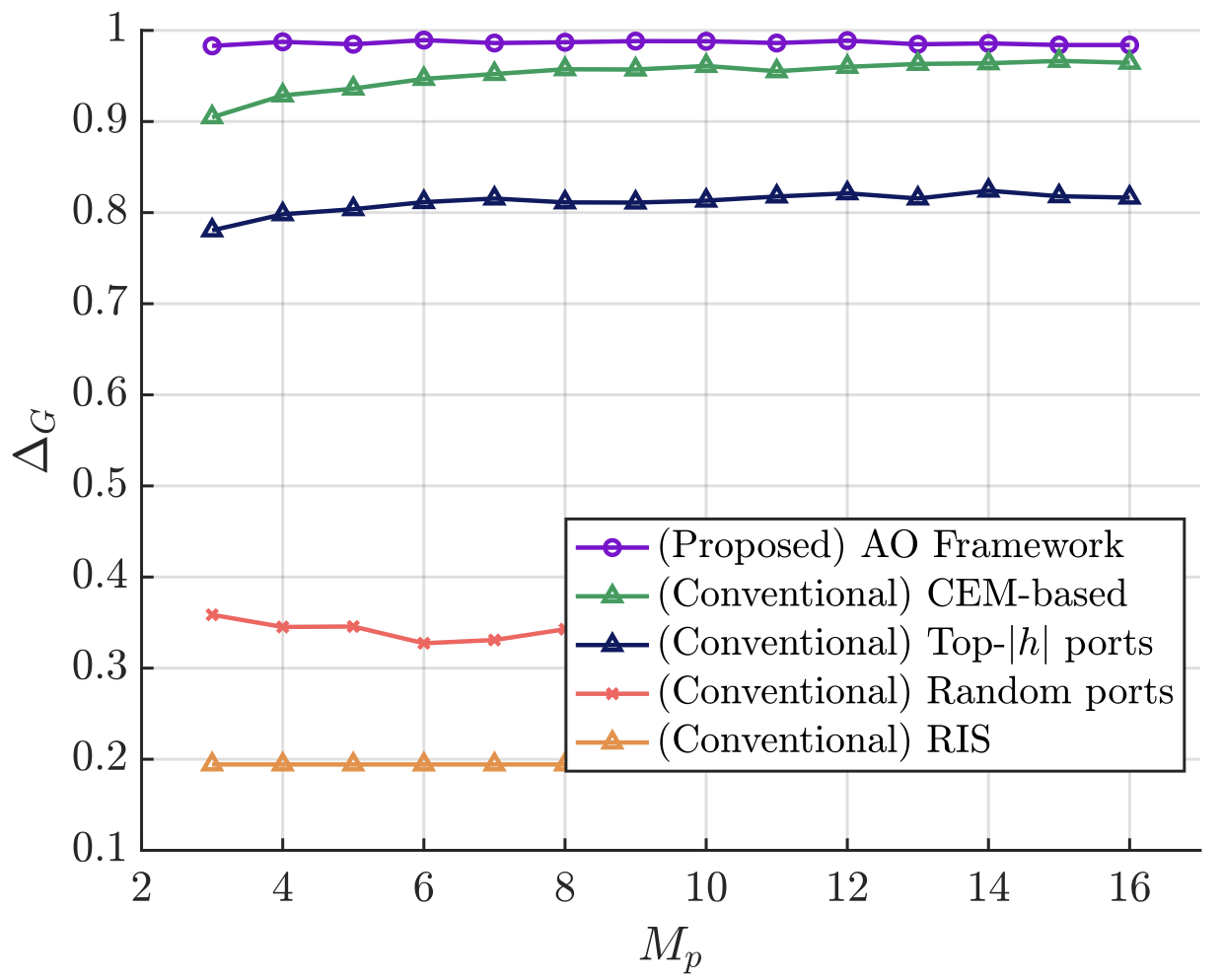}
        \label{fig_dmp}%
    }
    \caption{$\Delta_G$ according to (a) $M_o$ (general codebook) and (b) $M_{\mathrm p}$ ($M_{\mathrm p}$-gon codebook) with $M_o=4$.}
    \label{fig_3}
\end{figure}
%\begin{figure}[t]
%  \begin{center}
%    \includegraphics[width=0.5\columnwidth,keepaspectratio]{Fig/dmo}
%    \caption{$\Delta_G$ according to $M_o$ under general codebook.}
%    \label{fig_dmo}
%  \end{center}
%\end{figure}
%\begin{figure}[t]
%  \begin{center}
%    \includegraphics[width=0.5\columnwidth,keepaspectratio]{Fig/dmp}
%    \caption{$\Delta_G$ according to $M_{\mathrm p}$ under $M_{\mathrm p}$-gon codebook.}
%    \label{fig_dmp}
%  \end{center}
%\end{figure}

{Fig.~\ref{fig_dmo} shows $\Delta_G$ versus $M_o$ for general finite codebooks. The proposed framework achieves values of $\Delta_G$ close to one even for moderate $M_o$, indicating that it attains near-optimal performance compared to exhaustive search. This result confirms that the MRT and the Minkowski-geometry-based formulation effectively captures the essential structure of the optimal FRIS configuration. The CEM-based approach also achieves relatively high $\Delta_G$ values, but remains slightly below the proposed framework due to its sampling-based nature, which may not fully exploit the underlying geometric structure and can introduce variability in the solution quality. In contrast, the Top-$|h|$ and random-port schemes exhibit significantly lower $\Delta_G$, as their heuristic port-selection strategies fail to account for the joint geometric contribution of multiple ports. The conventional RIS baseline yields the lowest $\Delta_G$, since it operates with fixed element locations and thus lacks the position reconfigurability required to approach the optimal FRIS configuration.}

{Fig.~\ref{fig_dmp} illustrates $\Delta_G$ as a function of $M_p$ under a regular $M_p$-gon codebook. As $M_p$ increases, $\Delta_G$ rapidly approaches one, demonstrating that the proposed framework converges to the globally optimal solution as the phase resolution improves. Notably, near-optimal performance is achieved even for relatively small $M_p$, highlighting the robustness of the proposed design to coarse phase quantization. This observation confirms that the proposed framework can effectively operate under practical hardware constraints with limited phase resolution, thereby justifying the adoption of a codebook-based design that enables practical implementation while incurring negligible performance loss compared to optimum.}

%\begin{figure}[t]
%  \begin{center}
%    \includegraphics[width=0.4\columnwidth,keepaspectratio]{Fig/rk}
%    \caption{$\mathbb E[R]$ according to $K$ under general codebook.}
%    \label{fig_rk}
%  \end{center}
%\end{figure}

\begin{figure}[t]
    \centering
    \subfloat[]{%
        \includegraphics[width=0.24\textwidth]{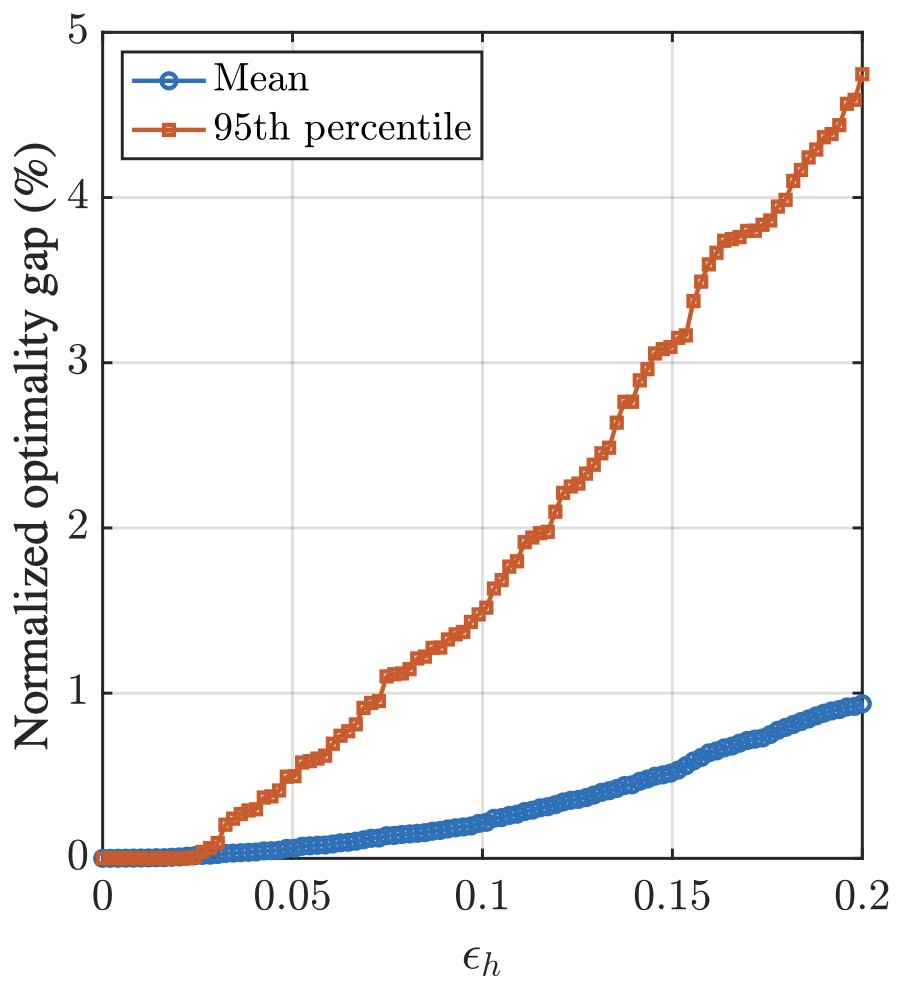}
        \label{fig_csi1}%
    }
    \subfloat[]{%
        \includegraphics[width=0.24\textwidth]{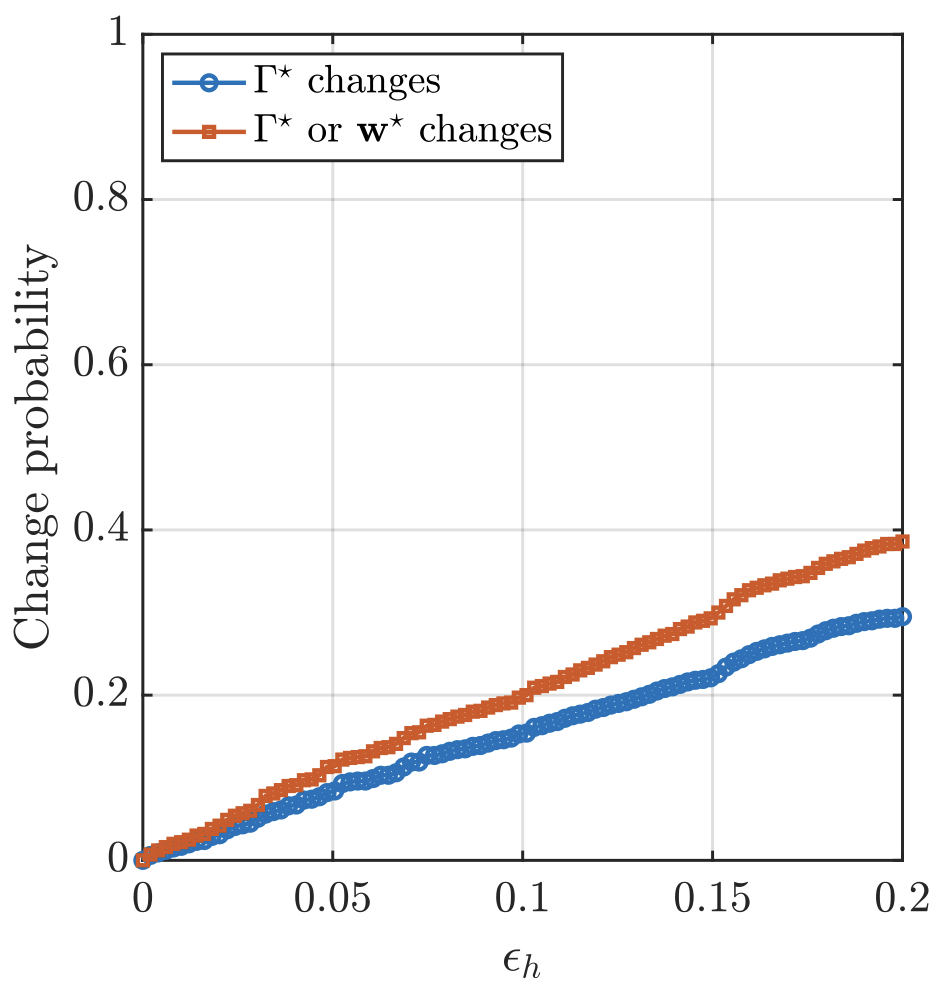}
        \label{fig_csi2}%
    }
    \caption{{(a) Normalized optimality gap of $|z|$ versus $\varepsilon_h$ and (b) probabilities of changes in FRIS configuration versus $\varepsilon_h$ under general codebook.}}
    \label{fig_csi}
\end{figure}
{Fig.~\ref{fig_csi} evaluates the robustness of the proposed FRIS design against CSI uncertainty as a function of \(\varepsilon_h\), where \(\varepsilon_d=\varepsilon_h\) with $M=9, M_o=3$, and ${|\mathcal W_m|=4}$. %For each channel realization, \(\mathbf w^\star\) denotes the optimal configuration obtained using true CSI, whereas \(\widehat{\mathbf w}^{\star}\) denotes the configuration optimized using imperfect CSI and subsequently evaluated over the true channel.
Fig.~\ref{fig_csi1} presents the normalized optimality gap with respect to $|z|$, defined as $\Delta_{\mathrm{norm}}
\triangleq
\frac{
|z(\mathbf w^\star)|
-
|z(\widehat{\mathbf w}^{\star})|
}{
|z(\mathbf w^\star)|
}$. Herein, both the mean and 95th-percentile gaps increase with \(\varepsilon_h\), reflecting the increasing mismatch between the true and estimated channel coefficients. Nevertheless, even at \(\varepsilon_h=0.2\), the mean optimality gap and 95th percentile remain below approximately 1 and 5\%, respectively. This demonstrates that the configuration optimized using imperfect CSI generally achieves a true beamforming gain close to the true-CSI optimum. Fig.~\ref{fig_csi2} reports the probability that imperfect CSI changes \(\Gamma^\star\) and the probability that it changes either \(\Gamma^\star\) or $\mathbf w^\star$. Both probabilities increase with \(\varepsilon_h\), while the overall-configuration change probability is consistently higher since the phase assignment may change even when the same ports remain selected. Nevertheless, these configuration changes do not translate into a comparable performance degradation, as shown in Fig.~\ref{fig_csi1}.\footnote{{Importantly, comparison in Fig.~\ref{fig_csi} shows that a change in the selected configuration does not necessarily cause a substantial performance loss. Although imperfect CSI changes $\Gamma^\star$ or $\mathbf w^\star$ in a non-negligible fraction of the realizations, the corresponding normalized optimality gap remains small. This indicates that CSI uncertainty frequently changes the ordering among several nearly optimal FRIS configurations rather than selecting a configuration with significantly inferior true-channel performance.}}}

{Overall, the proposed framework achieves superior performance by explicitly accounting for coherent combining through directional scoring and support-function maximization, rather than relying on magnitude-only or heuristic selection. While CEM provides an effective stochastic framework for jointly optimizing port selection and discrete phase shifts with comparable performance, the proposed framework offers several additional notable advantages. First, by leveraging the Minkowski-geometry-based formulation, the proposed method directly exploits the underlying structure of the FRIS configuration and yields a deterministic solution, thereby avoiding the variability and potential suboptimal convergence associated with sampling-based updates. Second, CEM incurs a complexity of $\mathcal{O}(A(M^2+2M+2M_o))$~\cite{FRISonoff} with $A\propto M$~\cite{CEM} is the number of samples, which effectively leads to a cubic scaling in $M$, in addition to requiring multiple iterations for convergence. In contrast, the proposed framework admits a structured implementation with {near-quadratic complexity of $\mathcal O(M^2M_p\log M)$ per AO iteration, as analyzed in Section~\ref{ccmp}}, resulting in significantly improved scalability. %Finally, CEM in~\cite{FRISonoff} is primarily tailored to single-user settings and does not readily extend to multi-user scenarios, whereas the proposed framework naturally accommodates multi-user extensions in Section~\ref{sec:multiuser_extension} through its geometry-driven formulation and structured design.
These advantages highlight the robustness and efficiency of the proposed framework.}

%{Overall, while the CEM-based approach provides an effective stochastic framework for jointly optimizing port selection and discrete phase shifts with comparable performance, the proposed framework offers several additional notable advantages. First, by leveraging the Minkowski-geometry-based formulation, the proposed method directly exploits the underlying structure of the FRIS configuration and yields a deterministic solution, thereby avoiding the variability and potential suboptimal convergence associated with sampling-based updates. Second, in terms of computational efficiency, the CEM-based method incurs a complexity of $\mathcal{O}(A(M^2+2M+2M_o))$~\cite{FRISonoff} with $A\propto M$~\cite{CEM} is the number of samples, which effectively leads to a cubic scaling in $M$, in addition to requiring multiple iterations for convergence. In contrast, the proposed framework admits a structured implementation with quadratic complexity, resulting in significantly improved scalability. Finally, the CEM-based approach in~\cite{FRISonoff} is primarily tailored to single-user settings and does not readily extend to multi-user scenarios, whereas the proposed framework naturally accommodates multi-user extensions in Section~\ref{sec:multiuser_extension} through its geometry-driven formulation and structured coefficient design. These advantages highlight the robustness, efficiency, and broader applicability of the proposed method.}

\section{Conclusion}
This paper presented a Minkowski-geometry-based framework for joint beamforming and FRIS configuration under practical port-selection and finite-resolution phase constraints. By representing the feasible reflected-signal set as a union of Minkowski sums over port subsets, the FRIS subproblem for a given beamformer was reduced to a one-dimensional directional search with Top-$M_o$ port selection, whose global optimum is attained over finitely many directions. For regular $M_p$-gon codebooks, closed-form scores further enabled an angular-sweep algorithm with near-quadratic complexity. Combined with closed-form MRT beamforming, the resulting AO framework converges monotonically within finitely many iterations, and the performance loss of the FRIS configuration under imperfect CSI was explicitly bounded. Simulations demonstrated near-optimal performance, consistent gains over benchmarks, robustness to CSI uncertainty, and rapid convergence.
%This paper presented a Minkowski-geometry-based framework for joint beamforming and FRIS configuration under practical port-selection and finite-resolution phase constraints. Using support functions, the FRIS subproblem was reduced to a one-dimensional directional optimization with constructive port selection and phase assignment, and further to a finite search over critical angles for regular $M_p$-gon codebooks. Combined with closed-form MRT beamforming, the proposed framework achieves convergence with low complexity. Simulations demonstrated near-optimal performance, consistent gains over conventional benchmarks, robustness to CSI uncertainty, and rapid convergence, highlighting the effectiveness of geometric structure for practical and scalable FRIS design.

\bibliographystyle{IEEEtran}
\bibliography{IEEEexample}

\end{document}